\documentclass[a4paper,UKenglish]{lipics-v2016}

\usepackage{microtype}%

\title{Dynamic Complexity under Definable Changes\footnote{The authors acknowledge the financial support by DFG grant SCHW 678/6-2. The third author thanks the Simons Institute for the Theory of Computing for hosting him and providing excellent research conditions.}}

\author{Thomas Schwentick}
\author{Nils Vortmeier}
\author{Thomas Zeume}

\affil{TU Dortmund University\\
  Germany\\
  \texttt{\{thomas.schwentick, nils.vortmeier, thomas.zeume\}@tu-dortmund.de}}

\authorrunning{T. Schwentick, N. Vortmeier, T. Zeume} %

\Copyright{Thomas Schwentick, Nils Vortmeier and Thomas Zeume}%

\subjclass{F.4.1. Mathematical Logic}%
\keywords{Dynamic descriptive complexity, SQL updates, dynamic programs}%
\ArticleNo{1}

\usepackage{mathtools}

\usepackage[pagewise]{lineno}
\usepackage{xspace}
\usepackage{ifmtarg}
\usepackage{paralist}
\usepackage{tikz}
\usetikzlibrary{arrows}

\newif\ifcomments
\newif\ifchanges
\commentsfalse\changesfalse
\commentstrue\changestrue %

\newcommand  {\myclass} [1]  {\ensuremath{\textsc{#1}}}

\newcommand{\StaClass}[1]{\myclass{#1}\xspace}

\newcommand{\DynClass}[1]{\myclass{Dyn#1}\xspace}

\newcommand  {\myproblem} [1] {\textsc{#1}}

\newcommand  {\problem}[1] {\myproblem{#1}}

\newcommand     {\AC}   {\myclass{AC}}

\newcommand{\FO}{\StaClass{FO}}
\newcommand{\MSO}[1][\quant]{\StaClass{MSO}}

\newcommand{\EFO}{\QFO[\exists^*]}
\newcommand{\AFO}{\QFO[\forall^*]}

\newcommand{\CQ}[1][]{\StaClass{CQ}}
\newcommand{\UCQ}[1][]{\StaClass{UCQ}}
\newcommand{\CQneg}[1][]{\StaClass{CQ\ensuremath{^{\mneg}}}}
\newcommand{\UCQneg}[1][]{\StaClass{UCQ\ensuremath{^{\mneg}}}}

\newcommand{\mneg}{\neg} %

\newcommand{\DynProp}{\DynClass{Prop}}

\newcommand{\DynFO}{\DynClass{FO}}

\newcommand{\Emptiness}[1][]{\problem{Emptiness}\ifthenelse{\equal{#1}{}}{}{(#1)}\xspace}
\newcommand{\Consistency}[1][]{\problem{Consistency}\ifthenelse{\equal{#1}{}}{}{(#1)}\xspace}
\newcommand{\HI}[1][]{\problem{HistoryIndependence}\ifthenelse{\equal{#1}{}}{}{(#1)}\xspace}

\newcommand{\mtext}[1]{\textsc{#1}}

\providecommand {\calA}      {{\mathcal A}\xspace}
\providecommand {\calB}      {{\mathcal B}\xspace}

\providecommand {\calD}      {{\mathcal D}\xspace}

\providecommand {\calI}      {{\mathcal I}\xspace}

\providecommand {\calP}      {{\mathcal P}\xspace}

\providecommand {\calS}      {{\mathcal S}\xspace}

\newcommand{\N}{\ensuremath{\mathbb{N}}}

\newcommand{\bigO}{\ensuremath{\mathcal{O}}}
\newcommand{\tpl}{\bar}

\newcommand{\df}{\ensuremath{\mathrel{\smash{\stackrel{\scriptscriptstyle{
    \text{def}}}{=}}}} \;}

\makeatletter %
\newcommand{\auxramsey}[4]{
  \@ifmtarg{#1}{
    \@ifmtarg{#4}{
      \ensuremath{R(#2; #3)}
    }{
      \ensuremath{R^#4(#2; #3)}
    }
   }{
    \@ifmtarg{#4}{
      \ensuremath{R_{#1}(#2; #3)}
    }{
      \ensuremath{R^#4_{#1}(#2; #3)}
    }
  }
}

   \theoremstyle{plain}

    \theoremstyle{definition}
    \newenvironment{proofof}[1]{\begin{proof}[Proof (of #1).]}{\end{proof}}
\newcommand{\arity}{\ensuremath{\text{Ar}}}

\newcommand{\schema}{\ensuremath{\tau}\xspace}

\newcommand{\relSchema}{\schema_{\text{rel}}}

\newcommand{\conSchema}{\schema_{\text{const}}}

\newcommand{\struc}{\calS}

\newcommand{\quant}{\mathbb{Q}}

\newcommand{\db}{\ensuremath{\calD}\xspace}
\newcommand{\inp}{\ensuremath{\calI}\xspace}
\newcommand{\aux}{\ensuremath{\calA}\xspace}

\newcommand{\domain}{\ensuremath{ D}\xspace}

\newcommand{\rel}[1]{\ensuremath{#1}}
\newcommand{\TC}{\rel{TC}}

\newcommand{\dom}{\ensuremath{\text{dom}}}

\newcommand{\state}{\ensuremath{\struc}\xspace}

\newcommand{\inpSchema}{\ensuremath{\schema_{\text{in}}}\xspace}
\newcommand{\auxSchema}{\ensuremath{\schema_{\text{aux}}}\xspace}

\newcommand{\prog}{\ensuremath{\calP}\xspace}

\newcommand{\updateDB}[2]{\ensuremath{#1(#2)}}
\newcommand{\updateState}[3]{\ensuremath{#1_{#2}(#3)}}
\newcommand{\uf}[4]{
  \@ifmtarg{#4}{
    \ensuremath{\phi^{#1}_{#2}(#3)}
   }{
    \ensuremath{\phi^{#1}_{#2}(#3; #4)}
  }
}
\newcommand{\huf}[4]{
  \@ifmtarg{#4}{
    \ensuremath{\widehat{\phi}^{#1}_{#2}(#3)}
   }{
    \ensuremath{\widehat{\phi}^{#1}_{#2}(#3; #4)}
  }
}

\newcommand{\ufb}[4]{
  \@ifmtarg{#4}{
    \ensuremath{\psi^{#1}_{#2}(#3)}
   }{
    \ensuremath{\psi^{#1}_{#2}(#3; #4)}
  }
}

  \makeatletter %
  \newcommand{\ufsubstitute}[5]{
    \@ifmtarg{#5}{
      \ensuremath{\phi^{#2}_{#3}[#1](#4)}
    }{
      \ensuremath{\phi^{#2}_{#3}[#1](#4; #5)}
    }
  }

  \makeatletter %
  
  \makeatletter %

\newcommand{\ut}[4]{
  \@ifmtarg{#4}{
    \ensuremath{t^{#1}_{#2}(#3)}
   }{
    \ensuremath{t^{#1}_{#2}(#3; #4)}
  }
}

\newcommand{\ite}[3]{
  \@ifmtarg{#1}{
    \ensuremath{\mtext{ITE}}
   }{
    \mtext{ITE}(#1,#2,#3)  
  }
}

\newcommand{\mf}[3]{
  \@ifmtarg{#3}{
    \ensuremath{\mu_{#1}(#2)}
   }{
    \ensuremath{\mu_{#1}(#2; #3)}
  }
}

\newcommand{\mfos}[4]{
  \@ifmtarg{#4}{
    \ensuremath{{#1}_{#2}(#3)}
   }{
    \ensuremath{{#1}_{#2}(#3; #4)}
  }
}

  \newcommand{\insertion}[2]{\textbf{insert}\ #2\ \textbf{into}\ #1}
  \newcommand{\deletion}[2]{\textbf{delete}\ #2\ \textbf{from}\ #1}
  \newcommand{\replace}[2]{\textbf{replace}\ #1\ \textbf{by}\ #2}

  \newcommand{\changeRule}[3]{
          \textbf{on change} #1 \textbf{update}  #2  \textbf{as}  #3
  }

\providecommand{\nc}{\newcommand}

\ifcomments
\nc{\commentbox}[1]{\noindent\framebox{\parbox{0.98\linewidth}{#1}}}
\nc{\todo}[1]{\ \\ {\color{red} \fbox{\parbox{0.98\linewidth}{{\sc
          ToDo}:\\  #1}}}}

\newcommand{\acomment}[2]{\ \\ \fbox{\parbox{0.98\linewidth}{{\sc #1}: #2}}}
\newcommand{\mcomment}[2]{{\color{blue}(#1)}\footnote{#1: #2}} %
\else
\nc{\commentbox}[1]{}
\newcommand{\mcomment}[2]{}
\newcommand{\acomment}[2]{}
\fi

\ifchanges

\newcommand{\loldnew}[3]{\commentbox{{\textcolor{blue}{\setlength{\fboxsep}{1pt}\fbox{\small
          #1}}} \textcolor{red}{\footnotesize #2}}
  \textcolor{blue}{#3}}
\setul{}{0.2mm}
\setstcolor{red}
\newcommand{\oldnew}[3]{{\textcolor{blue}{\setlength{\fboxsep}{1pt}\fbox{\small
        #1}}} \st{\footnotesize #2} \textcolor{blue}{#3}}

\else
\newcommand{\loldnew}[3]{#3}
\newcommand{\oldnew}[3]{#3}
\fi

 \nc{\tzm}[1]{\mcomment{TZ}{#1}}
 \nc{\tsm}[1]{\mcomment{TS}{#1}}
 \nc{\nilsm}[1]{\mcomment{NV}{#1}}
 
 \nc{\tz}[1]{\acomment{TZ}{#1}}
 \nc{\thz}[1]{\acomment{TZ}{#1}}
 \nc{\ts}[1]{\acomment{TS}{#1}}
 \nc{\nils}[1]{\acomment{NV}{#1}}

\nc{\tzon}[2][]{\oldnew{TZ}{#1}{#2}} 
\nc{\tson}[2][]{\oldnew{TS}{#1}{#2}}
\nc{\nilson}[2][]{\oldnew{NV}{#1}{#2}}

\nc{\tzlon}[2][]{\loldnew{TZ}{#1}{#2}} 
\nc{\tslon}[2][]{\loldnew{TS}{#1}{#2}}
\nc{\nilslon}[2][]{\loldnew{NV}{#1}{#2}}

\newcommand{\qr}{\ensuremath{\text{qd}}} %
\newcommand{\bd}{\ensuremath{\text{bd}}} %
\newcommand{\id}{\ensuremath{\text{id}}} %
\newcommand{\liff}{\ensuremath{\leftrightarrow}}
\newcommand{\Sigmae}{\ensuremath{\Sigma_\epsilon}}

\renewcommand{\EFO}{\ensuremath{\Sigma_1}}
\renewcommand{\AFO}{\ensuremath{\Pi_1}}

\newcommand{\Qreach}{\ensuremath{q_{\text{Reach}}}}
\usetikzlibrary{arrows, shapes, calc}%

\tikzstyle{mnode}=[
  circle,
  fill=\mnodefillcolor, 
  draw=\mnodedrawcolor,
  minimum size=6pt, 
  inner sep=0pt
]

\tikzstyle{mnodeinvisible}=[
  minimum size=6pt, 
  inner sep=0pt
]

\tikzstyle{invisible}=[
  minimum size=0pt, 
  inner sep=0pt
]

\tikzstyle{invisiblel}=[
  minimum size=10pt, 
  inner sep=0pt
]

\tikzstyle{invisibleEdge}=[
  transparent
]

\tikzstyle{nameNode}=[
  font=\scriptsize
]

\tikzstyle{namingNode}=[
  font=\normalsize
]

\tikzstyle{mEdge}=[
  -latex', %
  thick, 
  shorten >=3pt, 
  shorten <=3pt,
  draw=black!80,
]

\tikzstyle{dDashedEdge}=[
  -latex', %
  thick, 
  shorten >=3pt, 
  shorten <=3pt,
  draw=black!80,
  dashed
]

\tikzstyle{dEdge}=[
  -latex', %
  thick, 
  shorten >=1pt, 
  shorten <=1pt,
  draw=black!80,
]

\tikzstyle{dhEdge}=[
  -latex', %
  thick, 
  shorten >=3pt, 
  shorten <=3pt,
  draw=black!80,
]
\tikzstyle{uEdge}=[
  thick, 
  shorten >=3pt, 
  shorten <=3pt,
  draw=black!80,
]
\tikzstyle{uhEdge}=[
  thick, 
  shorten >=3pt, 
  shorten <=3pt,
  draw=black!80,
]

\tikzstyle{cEdge}=[
  ultra thick, 
  shorten >=3pt, 
  shorten <=3pt,
  draw=black!80,
]

\tikzstyle{dotsEdge}=[
  very thick, 
  loosely dotted, 
  shorten >=7pt, 
  shorten <=7pt
]

\tikzstyle{class rectangle}=[
  draw=black,
  inner sep=0.2cm,
  rounded corners=5pt,
  thick
]

\tikzstyle{mline}=[
  draw=black,
  inner sep=0.2cm,
  rounded corners=5pt,
  thick
]

\tikzstyle{mainclass rectangle}=[
  draw=blue,
  inner sep=0.2cm,
  rounded corners=5pt,
  very thick
]

\newcommand{\mnodedrawcolor}{black!80}
\newcommand{\mnodefillcolor}{black!40}

\pgfdeclarelayer{background}
\pgfdeclarelayer{substructure}
\pgfdeclarelayer{edges}
\pgfdeclarelayer{foreground}
\pgfsetlayers{background,substructure,edges,main,foreground}

\tikzstyle{background rectangle}=[
  fill=black!10,
  draw=black!20,
  line width = 5pt,
  inner sep=0.4cm,
  outer sep=0.4cm,
  rounded corners=5pt
]

\begin{document}
  \maketitle
  \begin{abstract}
This paper studies dynamic complexity under definable change
operations in the DynFO framework by Patnaik and Immerman.
It is shown that for changes definable by parameter-free
    first-order formulas, all (uniform) $\AC^1$ queries can be maintained  by
    first-order dynamic programs.  Furthermore, many maintenance
    results for single-tuple changes are extended to more powerful
    change operations: (1) The reachability query for undirected graphs is first-order
    maintainable under single tuple changes and first-order defined
    insertions, likewise the reachability query for directed acyclic graphs
    under quantifier-free insertions. (2) Context-free languages are first-order
    maintainable under $\EFO$-defined changes. %
    These results are complemented by several
    inexpressibility results, for example, that the reachability query cannot be maintained by quantifier-free programs under definable, quantifier-free deletions.
  \end{abstract}

  \section{Introduction}\label{section:introduction}

In the setting of \emph{Dynamic Complexity}, a database \db is being changed and an update program \prog tries to answer a standing query $q$ after each change. The program usually consists of logical formulas which can make use of additional, \emph{auxiliary} relations which in turn need to be updated after each change. Dynamic Complexity can be seen as a logic-based counterpart of \emph{Dynamic Algorithms}, where algorithms use auxiliary data structures to keep track of properties of structures like graphs under change operations. The Dynamic Complexity framework was introduced in \cite{PatnaikI97} and a similar framework, FOIES, in \cite{DongST95}. 

In Dynamic Complexity, one usually allows first-order logic formulas as update mechanism for the auxiliary relations. This is in line with the database-oriented framework, since first-order logic correspond to database languages  like relational algebra. Just as in Dynamic Algorithms, for most investigations the possible change operations are limited to insertions and deletions of single tuples. The class of queries maintainable in this fashion is called \DynFO.
 This line of research has seen recent progress, particular with respect to the question whether the reachability query can be maintained in \DynFO for directed graphs \cite{DattaHK14,DattaKMSZ15}.%

Although the restriction to single-tuple changes can be justified by the need to concentrate on the basic phenomena of dynamic maintainability of queries, it is also clear that from a more practical perspective one would be interested in more complex change operations at a time. One  approach is to specify changes by ``$\Delta$-relations'', e.g., by sets of tuples to be inserted or deleted. This is basically the viewpoint of \emph{Incremental View Maintenance} (see for example \cite{GuptaMS93}). However, it is clear that arbitrary $\Delta$-relations can make the auxiliary relations useless. 

In this work, we consider a different extension of the single-tuple-change paradigm that is inspired by SQL update queries (for a theoretical view at SQL updates we refer to \cite{AmelootBW13}). We model such queries by \emph{replacement queries} which can modify several relations at a time by first-order formulas that can use tuples of elements as parameters. Similar but slightly weaker frameworks were introduced in \cite{HesseI02, WeberS07}, but these papers did not study maintainability under such complex changes. %

\subparagraph*{Contributions}

The generalized setting yields a huge range of research questions, e.g., all previously studied questions in Dynamic Complexity in combination with replacement queries of varying expressiveness, and this paper can only start to investigate a few of them.

We are mainly interested in positive results. In Section~\ref{section:reach} we study first-order definable \emph{insertion} queries (supplementing the single tuple changes). It turns out that the reachability query can still be maintained in \DynFO for undirected graphs under first-order definable insertions (Theorem~\ref{theorem:tc:undirected:insertion}) and for directed acyclic graphs under quantifier-free insertions (Theorem~\ref{theorem:tc:acyclic:quantifierfree}). In Section~\ref{section:ac1}, we investigate parameter-free replacement queries. We show that \emph{all} queries that can be expressed in uniform $\AC^1$ (and thus all queries that can be computed with logarithmic space) can be maintained in \DynFO under first-order definable parameter-free replacement queries (Theorem~\ref{theorem:ac1}). In Section~\ref{section:languages}, we show that many maintainability results for formal languages \cite{PatnaikI97,GeladeMS12} carry over to quantifier-free or \EFO-definable replacement queries (Theorems~\ref{theorem:languages:DynProp} and~\ref{theorem:languages:EFO}). 

It is notoriously difficult to prove inexpressibility results in Dynamic Complexity. One would expect that allowing more general change operations simplifies such results. In Section~\ref{section:lower}, we confirm this intuition to some extent and present cases where general replacement queries disable certain kinds of update programs to maintain queries that are maintainable under single-tuple changes.

\subparagraph*{Related work}
In addition to the related work mentioned already above, several other prior results for Dynamic Complexity under more general changes have been obtained. The reachability query for directed graphs has been studied  under deletions of sets of edges and nodes that form an anti-chain in \cite{DongP97} and under insertions of sets of tuples that are cartesian-closed in \cite{DongST95}. Hesse observed that the maintenance procedure for this query under single tuple changes from \cite{DattaKMSZ15} can deal with the replacement of the set of outgoing edges of a node (or, alternatively, the set of incoming edges). Edge contractions have been studied in \cite{Siebertz11}. Koch considered more general sets of changes in \cite{Koch10}, though only for non-recursive queries.

Implementations of work on Dynamic Complexity are reported in \cite{PangDR05} and \cite{Koch10}.

  \section{Preliminaries}\label{section:preliminaries}

As much of the original motivation for the investigation of dynamic
complexity came from incremental view maintenance (cf.\ \cite{DongT92,DongS93,PatnaikI97}),
it is common to consider logical structures as relational
databases and to use notation from relational databases.

A \emph{(relational) schema} $\schema$ consists of a set $\relSchema$
of relation symbols, accompanied by an arity function $\arity:
\relSchema \rightarrow \N$, and a set $\conSchema$ of constant
symbols. In this work, a \emph{domain} is a finite set. A
\emph{database} $\db$ over schema $\schema$ with domain $\domain$
assigns to every relation symbol $R \in \relSchema$ a relation of
arity~$\arity(R)$ over $\domain$  and to every constant symbol $c \in
\conSchema$ an element (called \emph{constant}) \mbox{from
  $\domain$}.   
A  $\schema$-\emph{structure} $\struc$ is a pair $(\domain, \db)$
where $\domain$ is a domain and $\db$ is a database with domain
$\domain$ over schema $\schema$. By $\dom(\struc)$ we refer to $\domain$.
For a relation symbol $R \in \schema$ and a constant symbol~$c \in \schema$ we denote by $R^\state$ and $c^\state$ the relation and constant, respectively, that are assigned to those symbols in~$\state$. A \emph{$k$-ary query} $q$ on $\schema$-structures is a mapping that assigns a subset of $\domain^k$ to every $\schema$-structure over \mbox{domain $\domain$} and is closed under isomorphisms.

We represent graphs as structures over a schema that contains a single binary relation $E$. The \emph{reachability query} $\Qreach$ maps graphs to their transitive closure relation.

In Section~\ref{section:languages} we consider databases that represent words over some  alphabet $\Sigma$. In a nutshell, the positions of a word correspond to elements of the domain and the letters at positions are indicated by unary relations. More formally,
words are represented by databases with an immutable linear order on their domain and one unary relation $R_\sigma$ for every $\sigma\in\Sigma$.
For simplicity, we always assume that the domain of such a database is of the form $\{1,\ldots,n\}$ and the linear order ist just the natural order. At any point in time, an element of the domain is allowed to be in at most\footnote{There are ways to get rid of this requirement, but we keep it for simplicity.} one relation $R_\sigma$. However, elements need not to be in any relation $R_\sigma$ and, in this case, they do not correspond to a position with a symbol but rather to the empty word $\epsilon$. Thus, we first associate with every position $i$ an element $w_i\in\Sigmae$, where by $\Sigmae$ we denote the set~$\Sigma\cup\{\epsilon\}$, and say that the database represents the string $w=w_1\cdots w_n$. As an example, the database with domain $\{1,2,3,4,5\}$ and $R_a=\{2,4\}, R_b=\{1\}$ represents the string $baa$. As a further convenience, we assume that databases have constants $\min$ and $\max$ that represent the smallest and the largest element, $1$ and $n$, respectively.\footnote{This assumption can be avoided by, e.g., using additional prefix and suffix relations in the proof of Theorem~\ref{theorem:languages:DynProp}, in the spirit of \cite{GeladeMS12}.} We will not allow the linear order, $\min$ or $\max$ to be modified by change operations.
In this paper, we will rarely distinguish between a database and the string it represents. %

We use several notions from finite model theory (see, e.g., \cite{Libkin04}). By $\qr(\varphi)$, we denote the \emph{quantifier-rank} of a first-order formula $\varphi$, that is, its maximum nesting depth of quantifiers. We denote the set of \emph{rank-$k$ types} of tuples of arity $\ell$ by $\FO[k,\ell]$ (cf. \cite[Definition 3.14]{Libkin04}). The existential fragment of first-order logic is denoted by $\EFO$.

  \section{Dynamic Programs with Complex Changes} \label{section:framework}

In this section we lift the definitions from \cite{SchwentickZ16} to more general change operations.  We first define (general) change operations, then we adapt the definition of dynamic programs presented in \cite{SchwentickZ16} to those  more complex changes.

\subparagraph*{Change Operations} The change operations that we consider in this paper are based on queries. In their most general form, they can modify a given database over schema $\schema$ by replacing some of its relations with the results of first-order-defined  queries on the database. These queries are allowed to use parameters.
 
To this end, a \emph{replacement rule} $\rho_R$ for relation $R$ is of the form $\replace{R}{\mf{R}{\tpl p}{\tpl {x}}}$. Here, $R$ is a relation symbol and $\mf{R}{\tpl p}{\tpl x}$ is a first-order formula over $\schema$, where the tuple $\tpl x$ has the same arity as $R$ and $\tpl p$ is another tuple of variables, called the \emph{parameter tuple}.  A \emph{replacement query} $\rho(\tpl p)$ is a set of replacement rules for distinct relations with the same parameter tuple $\tpl p$. In the case of replacement queries $\rho$ that consist of a single replacement rule, we usually do not distinguish between $\rho$ and its single replacement formula $\mu_R$.

For a database $\db$, a change operation $\delta=(\rho,\tpl a)$ consists of a replacement query and a tuple of elements of (the domain of) \db with the same arity as the parameter tuple of~$\rho$. We often use the more concise notation $\rho(\tpl a)$ and refer to change operations simply as \emph{changes}.

The result $\updateDB{\delta}{\db}$ of an application of a change operation $\delta = (\rho,\tpl a)$ to a database $\db$ is defined in a straightforward way: each relation $R$ in \db, for which there is a replacement rule $\rho_R$ in $\rho$, is replaced by the relation resulting from evaluating $\mu_R$, that is, by $\{\tpl b\mid \db \models \mf{R}{\tpl a}{\tpl b}\}$.

If a replacement query has no parameters we say that it is \emph{parameter-free}.

\begin{example}
  \begin{enumerate}[(a)]
    \item As a first example, we consider directed graph structures, that is, structures with a single binary relation $E$. Let, for some graph $G$,  $\delta_1=(\rho_1,u)$ be the change operation  with replacement query $\mf{E}{p}{x,y} = E(x,y) \vee (x = p)$ and node $u$. Then, in $\delta_1(G)$, there is an edge from $u$ to every node of $G$.
    \item We recall that words over the alphabet $\Sigma = \{a, b , c\}$ are represented by databases with a linear order on their domain and one unary relation $R_\sigma$ for every $\sigma \in \Sigma$. Let \db be a database representing a word $w$ and $i$ an element of $\db$. Let $\delta_2=(\rho_2,i)$ be a change operation, where the replacement query $\rho_2$ consists of the rules 
    $\replace{R_a}{\mf{R_a}{p}{x}}$ and  $\replace{R_b}{\mf{R_b}{p}{x}}$
    with $\mf{R_a}{p}{x} = \big((x < p) \wedge R_b(x)\big) \vee \big(\neg (x < p) \wedge R_a(x) \big)$ and $\mf{R_b}{p}{x} = \big((x < p) \wedge R_a(x)\big) \vee \big(\neg (x < p) \wedge R_b(x) \big)$. Then, $\delta_2(\db)$ represents the word obtained from $w$ by swapping $a$ and $b$ symbols on all positions before $i$ and leaving all other positions unchanged.
  \end{enumerate}
\end{example}

Some of our investigations will focus on (syntactically) restricted replacement queries that either only remove or only insert tuples to relations. For an \emph{insertion rule} $\rho_R$, the  replacement formula $\mf{R}{\tpl p}{\tpl x}$ has the form $R(\tpl x) \lor \varphi_R$. Similarly, \emph{deletion rules} have replacement formulas $\mf{R}{\tpl p}{\tpl x}$ of the form $R(\tpl x) \land \varphi_R$. In \cite{AmelootBW13}, the change operations \emph{replace}, \emph{insert}, \emph{delete} and \emph{modify} have been studied, in particular with respect to their expressive power. These operations are captured by our change operations\footnote{In \cite{AmelootBW13} the domain of the database can be infinite.}.

Another syntactic restriction to be studied extensively in this work are the quantifier-free replacement queries, that allow only quantifier-free change formulas to be used. A special case of quantifier-free changes are the \emph{single tuple changes}. We refer by $\insertion{R}{\tpl p}$ to the insertion query $\replace{R}{\mf{R}{\tpl p}{\tpl x}}$, where $\mf{E}{\tpl p}{\tpl x} = R(\tpl x) \lor (\tpl p = \tpl x)$ and by  $\deletion{R}{\tpl p}$ to the deletion query $\replace{R}{\mf{R}{\tpl p}{\tpl x}}$, where \mbox{$\mf{R}{\tpl p}{\tpl x} = R(\tpl x) \land \neg(\tpl p = \tpl x)$}.  As mentioned before, single tuple changes are the best studied change operations in previous work on dynamic complexity. To emphasize the difference we sometimes refer to arbitrary (not single-tuple) change operations as \emph{complex changes}. For any schema $\tau$ we denote by $\Delta_\tau$ the set of single-tuple replacement queries for the relations (with symbols) in~$\tau$. In the case of graphs, we simply write $\Delta_E$. In case of strings over some alphabet $\Sigma$, we write $\Delta_\Sigma$.

\subparagraph*{Dynamic Programs}

We now introduce dynamic programs, closely following the exposition in \cite{SchwentickZ16}.
Inputs in dynamic complexity are represented as relational structures as defined in Section~\ref{section:preliminaries}.  The domain is fixed from the beginning, but the database in the initial structure is empty. This initially empty structure is then modified by a sequence of change operations.

The goal of a dynamic program is to answer a given query for the database that results from any change sequence. To this end, the program can use an auxiliary data structure represented by an auxiliary database over the same domain. Depending on the exact setting, the auxiliary
database might be initially empty or not.

A dynamic program $\prog$ operates on an \emph{input database} $\inp$ over a schema $\inpSchema$ and updates an \emph{auxiliary database} $\aux$ over a schema\footnote{To simplify the exposition, we will usually not mention schemas explicitly and always assume that all structures we talk about are compatible with respect to the schemas at hand.} $\auxSchema$, both sharing the same domain $D$ which is fixed during a computation. We call $(D,\inp, \aux)$ a \emph{state} and consider it as one relational structure. The relations of $\inp$ and $\aux$ are called \emph{input and auxiliary relations}, respectively.

A \emph{dynamic program} has a set of update rules that specify how auxiliary relations are updated after a change. An \emph{update rule} for updating an auxiliary relation $T$ after a replacement query $\rho(\tpl p)$ is of the form \changeRule{$\rho(\tpl p)$}{$T(\tpl x)$}{$\varphi_T(\tpl p, \tpl x)$}
where the \emph{update formula}  $\varphi_T$ is over $\inpSchema \cup \auxSchema$. 

The semantics of a dynamic program is as follows.  When a change operation $\delta = \rho(\tpl a)$ is applied to the input database $\inp$, then the new state $\state$ of $\prog$ is obtained by replacing the input database by $\delta(\inp)$ and by defining each auxiliary relation $T$ via $T \df \{\tpl b \mid (\inp, \aux) \models \varphi_T(\tpl a, \tpl b)\}$.  For a change operation $\delta$ we denote the updated state by $\prog_\delta(\state)$. For a sequence $\alpha=(\delta_1,\ldots,\delta_k)$ we write $\prog_\alpha(\state)$ for the state obtained after successively applying $\delta_1,\ldots,\delta_k$ to $\state$.

A \emph{dynamic query} is a tuple $(q, \Delta)$ where $q$ is a query over schema $\inpSchema$ and $\Delta$ is a set of replacement queries. The dynamic program $\prog$ \emph{maintains} a dynamic query $(q, \Delta)$ with $k$-ary $q$ if it has a $k$-ary auxiliary relation $Q$ that, after each change sequence over $\Delta$, contains the result of $q$ on the current input database. More precisely, for each non-empty\footnote{This technical restriction ensures that we can handle, e.g., Boolean queries with a yes-result on empty structures without initialization of the auxiliary relations. Alternatively, one could use an extra formula to compute the query result from the auxiliary (and input) structure.} sequence $\alpha$ of changes and each empty input structure $\inp_\emptyset$, relation $Q$ in $\prog_\alpha(\state_\emptyset)$ and $q(\alpha(\inp_\emptyset))$ coincide. Here, $\state_\emptyset=(\inp_\emptyset, \aux_\emptyset)$, where $\aux_\emptyset$ denotes the empty auxiliary structure over the domain of $\inp_\emptyset$.

The class of dynamic queries $(q, \Delta)$ that can be maintained by a dynamic program with update formulas from first-order logic is called $\DynFO$. We also say that the query $q$ can be maintained in $\DynFO$ under change operations $\Delta$. The class of dynamic queries maintainable by quantifier-free update formulas is called $\DynProp$. 

The following very simple example shows how the transitive closure of a directed graph subject to single edge insertions can be maintained in this set-up. 

\begin{example}\label{example:tcinsert}
Let $\Qreach$ be the reachability query that returns all pairs $(u,v)$ of a graph, for which there is a path from $u$ to $v$. The dynamic query $(\Qreach,\{\insertion{E}{\tpl p}\})$ can be maintained by a dynamic program that uses one auxiliary relation $T$, which always contains the transitive closure of the edge relation $E$. Its only update rule is given by the formula $\varphi_T(p_1,p_2;x,y)=T(x,y) \lor \big(T(x,p_1) \wedge T(p_2,y)\big)$. 
 \qed%
\end{example}

Our general framework follows \cite{PatnaikI97} and thus does not allow inserting new elements into or removing existing elements from the domain as in the FOIES framework \cite{DongST95}. The step from Dynamic Complexity to FOIES can be done by adding two more change operations,  $\text{add}(x)$ and $\text{remove}(x)$. Our results of Section~\ref{section:reach}  easily carry over, and those of Section~\ref{section:languages} carry over if, say, new elements are always added at the end of the string. Since $\text{add}(x)$ and $\text{remove}(x)$ have parameters, they do not quite fit into the parameter-free framework of Section~\ref{section:ac1}. However, Theorem~\ref{theorem:ac1} survives if parameter-free remove queries are allowed.

\subparagraph*{Complex Change Operations and Initialization of Dynamic Programs.}

In the presence of complex replacement queries, the initialization of the auxiliary relations requires some attention. In the original setting of Patnaik and Immerman, the input database is empty at the beginning, and the auxiliary relations are initialized by first-order formulas evaluated on this (empty) initial input database. Since tuples can be inserted only one-by-one, the auxiliary relations can be adapted slowly and it can be ensured that, e.g., always a linear order \cite{PatnaikI97} or arithmetic \cite{Etessami98} on the active domain is available. 

For complex changes, the situation is more challenging for a dynamic program: as an example, in the setting of strings, the first change could insert all positions of the domain into relation $R_a$ and thus let the database represent the word $a^n$, if $n$ is the size of the underlying domain. To enable the dynamic program to answer whether the string is in some language after this change, it needs some suitable (often: non-empty) initial values of the auxiliary relations. Since in this paper, we are mainly interested in the maintenance of queries and not so much in the specific complexity of the initialization, we do not define variants of \DynFO with different power of initialization, but rather  follow a pragmatic approach: whenever initialization is required, we say that the query can be maintained \emph{with suitable initialization} and specify in the context what is actually needed. In all cases, it is easy to see that the initialization of the auxiliary relations can be computed in polynomial time. 

An alternative approach would be to restrict the semantics of replacement queries to elements of the active domain of the current database and to allow the activation of elements only via tuple insertions.

  \section{Reachability and Definable Insertions}\label{section:reach}
In this section, we study the impact of first-order definable complex change operations on the (binary) reachability query $\Qreach$. We present positive cases, where previous maintainability results survive under such stronger change operations. Negative results, where such operations destroy previous maintainability results, are given in Section~\ref{section:lower}. 

In the classical \DynFO setting with single-tuple change operations it was shown early on that $\Qreach$ can be maintained in \DynFO for two important graph classes: undirected graphs and directed, acyclic graphs (dags). It turns out that these results still hold in the presence of complex \emph{insertions}: first-order insertions for undirected graphs and quantifier-free insertions for dags. In fact, in both cases basically the same auxiliary relations can be used as in the case of single-tuple changes. %

We first show that for undirected graphs, the reachability query can be maintained in \DynFO, for first-order insertions and the set $\Delta_E$ of single-edge insertions and deletions. We follow the convention from \cite{GradelS12} that modifications for undirected graphs are symmetric in the sense that if an edge $(a,b)$ is inserted then so is the edge $(b, a)$ (and likewise for deletions).

\begin{theorem}\label{theorem:tc:undirected:insertion}
Let $\Delta$ be a finite set of first-order insertion queries. Then $(\Qreach,\Delta\cup\Delta_{E})$ can be maintained in \DynFO for undirected graphs.
\end{theorem}

We use the approach for maintaining $\Qreach$ for undirected graphs under single-edge insertions and deletions from \cite[Theorem 4.1]{PatnaikI97} and maintain a spanning forest and (essentially) its transitive closure relation. The crucial observation for extending this approach to first-order insertions is that, after such an insertion, between each pair of nodes in a (new) connected component, there is a connecting path that uses only a bounded number of newly inserted edges. This allows the update of the spanning forest and its transitive closure in a first-order definable way.

The observation is stated more precisely next. For two connected nodes $u,v$ in a graph $G'=\delta(G)$ that is obtained from a graph $G$ by an insertion $\delta$, we define\footnote{Since $G$ and $\delta$ will be always clear from the context, we do not add them as parameters to this notation.} the \emph{bridge distance} $\bd(u,v)$ as the minimal number $d$, such that there is a path from $u$ to $v$ in $G'$ that uses $d$ edges that were newly inserted by $\delta$.

\begin{lemma}\label{lemma:tc:undirected:boundnewedges}
For each first-order insertion query  $\rho$ there is a constant $m \in \N$ such that for each undirected graph $G$, each change $\delta = \rho(\tpl a)$ and all nodes $u$ and $v$ of $G$ that are connected in $\delta(G)$ it holds $\bd(u,v)\le m$.
\end{lemma}
We informally refer to this property as the \emph{bridge boundedness property}.
\begin{proof}

The proof makes use of the result by Feferman-Vaught that the depth $k$ first-order type of the disjoint union of two structures is determined by the depth $k$ first-order types of these two structures \cite{Feferman57, FefermanV59} (see also \cite{Makowsky04}). 

Let $\mu(\tpl p;\tpl x)$ be the first-order formula underlying $\rho$ and $k$ its quantifier-rank. Let $\ell$ be the arity of $\tpl p$, $m'$ the number of $\FO[k,1]$-types of undirected graphs and $m=\ell+m'$.

Let $G$ be an undirected graph and let $\delta = \rho(\tpl a)$ for some tuple $\tpl a$ of nodes of $G$. 
Let $u,v$ be two nodes that are connected by some path $\pi$ of the form $u = w_0, w_1,\ldots, w_{r} = v$ in $\delta(G)$ with $q$ bridges, that is, edges that are not in $G$. Our goal is to show that there exists such a path with at most $m$ bridges. Thus, if $q\le m$, there is nothing to prove, so we assume $q>m$. It suffices to show that there is a path from $u$ to $v$ with fewer than $q$ bridges.
Let $(u_1,v_1),\ldots,(u_q,v_q)$ be the bridges in $\pi$. If for some $i$, the nodes $u_i$ and $v_i$ are in the same connected component of $G$ (before the application of $\delta$), we can replace the bridge $(u_i,v_i)$ by a path of ``old'' edges resulting in an overall path with $q-1$ bridges. Similarly, if $u_i$ and $u_j$  are in the same connected component of $G$, for some $i<j$, we can shortcut $\pi$ by a path from $u_i$ to $u_j$ inside $G$.
Therefore, we can assume that, for every $i$, the nodes $u_i$ and $v_i$ are in different connected components of $G$, and likewise $u_i$ and $u_j$ for $i<j$.

We show that in this case there are $i, j$ with $i < j$ such that $\mu$ defines an edge between $u_i$ and $v_j$, and therefore a path with fewer bridges can be constructed by shortcutting the path $\pi$ with the edge $(u_i, v_j)$.  By the choice of $m$ there must be two nodes $u_i$ and $u_j$, with $i<j$, in distinct connected components of $G$ that do not contain any element from $\tpl a$, such that $u_i$ and $u_j$ have the same $\FO[k,1]$-type in their respective connected components. By Feferman-Vaught, it follows that $(u_i,v_j,\tpl a)$ and $(u_j,v_j,\tpl a)$ have the same $\FO[k,\ell+2]$-types and therefore, since $\mu$ defines an edge between $u_j$ and $v_j$, it also defines one between $u_i$ and~$v_j$.
\end{proof}

\begin{proofof}{Theorem~\ref{theorem:tc:undirected:insertion}}
The dynamic program presented in \cite[Theorem 4.1]{PatnaikI97} maintains the transitive closure of undirected graphs under single-edge changes with the help of auxiliary relations $F$ and $PV$. 
The binary relation $F$ is a spanning forest of the input graph $G$ and $(u,v,w) \in PV$ means that $w$ is a node in the path from $u$ to $v$ in $F$. 
Observe that two nodes $u$ and $v$ are connected in an undirected graph if and only if $(u,u,v) \in PV$ holds. %

We show how to maintain the relation $F$ and $PV$ under \FO insertions. For the moment we assume a predefined linear order $\leq$ on the domain to be present.
Let $\rho$ be an insertion query and $m$ the bound on the number of bridges by Lemma \ref{lemma:tc:undirected:boundnewedges}.
Let $G$ be an undirected graph and $\delta = \rho(\tpl a)$ an insertion, $F$ a spanning forest of $G$ and $PV$ as described above.
We show how to \FO-define the auxiliary relations $F'$ and $PV'$ for the modified graph~$G'=\delta(G)$.

We first describe a strategy to define $F'$ and then argue that it can be implemented by a first-order formula. Let $C'$ be a (new) connected component in $\delta(G)$. We call the smallest node of $C'$ with respect to $\leq$ the \emph{queen} $u_0$ of $C'$. For each connected component $C$ of~$G$ that is a subgraph of $C'$, we define its \emph{queen level} as the (unique) number $\bd(u,u_0)$, for nodes $u\in C$. A bridge in $C'$ is inserted into $F'$ if for a connected component $C$ of~$G$ of some level $i$ it is the lexicographically smallest edge with respect to $\leq$ that connects $C$ with some component of level $i-1$. This clearly defines a spanning forest. The chosen edges can be defined by a first-order formula because, for each number $h$, there are formulas $\theta_h(x,y)$ expressing that $\bd(x,y)\le h$. 

Since the construction of $F'$ ensures that each path in $F'$ from a node to the queen of its connected component only contains at most $m$ new edges, and thus each path in $F'$ contains at most $2m$ new edges, it is straightforward to extend the update formula for $PV'$ from~\mbox{\cite[Theorem 4.1]{PatnaikI97}}.

It remains to show how the assumption of a predefined linear order can be eliminated. For a change sequence $\alpha$, we denote by $A_\alpha$ the set of parameters used in $\alpha$. When applying $\alpha$ to an initially empty graph, a linear order on $A_\alpha$ can be easily constructed as in the case of single-tuple changes \cite{Etessami98}. The remaining nodes in $V \setminus A_\alpha$ behave very similarly. More precisely,  one can show by induction on $|\alpha|$, that for all nodes $a\in V$ and $b,b'\in V-A_\alpha$ it holds $(a,b) \in E \Leftrightarrow (a,b') \in E$ (and likewise for $(b,a)$ and $(b',a)$).

The dynamic program for maintaining $\Qreach$ for undirected graphs now maintains the relations $F$ and $PV$ as described above, yet restricted to the induced (and ordered) subgraph of $G$ on~$A_\alpha$. The transitive closure can be defined from those relations and the edge relation by a simple case distinction.
\begin{itemize}\item 
Two nodes $a,a' \in A_\alpha$ are connected by a path if and only if there is a path from $a$
  to $a'$  in $F$ or if there are nodes $b, b' \in A_\alpha$ and a node $c \in V \setminus  A_\alpha$ such that there are paths from $a$ to $b$ and from
  $a'$  to $b'$
  in $F$ as well as edges $(b,
  c)$ and $(b', c)$. 
\item Two nodes $a \in A_\alpha$ and $b \in V \setminus A$ are connected by a path if and only if there is a node $a' \in A_\alpha$ such that there is a path from $a$ to
  $a'$  in  $F$
  and an edge $(a',b)$. 
\item Finally, two nodes $a, a' \in V \setminus
  A_\alpha$ are connected if and only if there is an edge $(a,
  a')$ or there is an edge $(a, b)$ for some $b \in
  A_\alpha$ (and therefore also an edge $(a', b)$).
\end{itemize}
\end{proofof}

Now we turn to the other graph class, acyclic graphs, for which \DynFO maintainability under complex insertions (and single-edge deletions) is preserved; albeit (we are able to show that) only for quantifier-free insertions. In \cite[Theorem 4.2]{PatnaikI97}, edge insertions are only allowed if they do not add cycles. Of course, given the transitive closure of the current edge relation it can be easily checked by a first-order formula (a \emph{guard}), whether a new edge closes a cycle. We will see that this is also possible for the complex insertions we consider.

\begin{theorem}\label{theorem:tc:acyclic:quantifierfree}
Let $\Delta$ be a finite set of quantifier-free insertion queries. Then $(\Qreach,\Delta\cup\Delta_{E})$ can be maintained in \DynFO for directed, acyclic graphs. Furthermore, for each quantifier-free insertion, there is a first-order guard which checks whether the insertion destroys the acyclicity of the graph.
\end{theorem}

As in the case of undirected graphs, the proof relies on a bridge boundedness property. This property allows extending the technique for maintaining the transitive closure relation of acyclic graphs under single tuple changes used in \cite{PatnaikI97} and \cite{DongS95} to quantifier-free insertions. As in  \cite{PatnaikI97} and \cite{DongS95} no further auxiliary relations besides the transitive closure relation are needed. In Section \ref{section:lower} we show that the transitive closure relation does not suffice for maintaining $\Qreach$ for acyclic graphs subjected to $\EFO$-definable insertions\footnote{Indeed, the graphs $G$ and $G'$ used in the proof of Theorem \ref{theorem:tc:general:insertion}(b) show that the following lemma fails already for $\EFO$-insertions. }.

In the following lemma, the bridge distance $\bd$ is defined just as above. However, we can no longer assume that bridges connect (formerly) different connected components, therefore the lemma only holds for quantifier-free insertions.
The proof can be found in the full version of this paper.

\begin{lemma}\label{lemma:tc:acyclic:boundnewedges}
For each quantifier-free insertion query  $\rho$ there is a constant $m \in \N$ such that for each directed, acyclic graph $G$ and each change $\delta = \rho(\tpl a)$ it holds that $\delta(G)$ has a cycle with at most $m$ bridges, or for all nodes $u$ and $v$ of $G$ with a path from $u$ to $v$ in $\delta(G)$,  it holds $\bd(u,v)\le m$.
\end{lemma}

\begin{proofof}{Theorem~\ref{theorem:tc:acyclic:quantifierfree}}
 In \cite[Theorem 4.2]{PatnaikI97} and \cite[Theorem 3.3]{DongS95}, dynamic programs are given that maintain the transitive closure of acyclic graphs under single-edge modification, using only the transitive closure as auxiliary relation. Thanks to Lemma~\ref{lemma:tc:acyclic:boundnewedges}, these programs can be easily extended. Indeed, since the number of bridges of cycles created by the insertion, and, if the graph remains acyclic, the bridge distance between two path-connected nodes are bounded by a constant, a guard formula and an update formula for the transitive closure can be constructed in a straightforward manner.
\end{proofof}

  \section{Parameter-free Changes}\label{section:ac1}
\newcommand{\Ind}[1]{\ensuremath{\text{IND}[#1]}}
\newcommand{\BIT}{\ensuremath{\text{BIT}}}

In this section we consider replacement queries without parameters on ordered databases. It turns out that in this case a large class of queries can be maintained in \DynFO: all queries that can be expressed in uniform $\AC^1$ and thus, in particular, all queries that can be answered in logarithmic space. 
This result exploits the fact that for a fixed set of replacement queries without parameters there is only a constant number of possible changes to a structure.

An \emph{ordered} database $\db$ contains a built-in linear order $\leq$ on its domain that is not modified by any changes. One might suspect that parameter-free replacement queries are not very powerful, especially when they are applied to the initially empty input database. However, thanks to the linear order, one can actually construct every finite graph with relatively simple replacement queries (and similarly for other kinds of databases). For instance, one can cycle through all pairs of nodes in lexicographic order. If $(u,v)$ is the current maximal pair, operation \emph{keep} can move to $(u,v+1)$ (inserting it into $E$) while leaving $(u,v)$ in $E$ and \emph{drop} can 
move to $(u,v+1)$ while taking $(u,v)$ out from $E$.

The update programs constructed in this section use, as additional auxiliary relation, a binary \BIT-relation containing all pairs $(i,j)$ of numbers, for which the $i$-th bit of the binary representation of $j$ is 1. Here, we identify elements of an ordered database with numbers. In the following, the minimal element with respect to $\leq$ is considered as $0$.

By $\AC^1$ we denote the class of problems that can be decided by a uniform\footnote{for concreteness: first-order uniform \cite{ImmermanDC}} family of circuits of ``and'', ``or'' and ``not'' gates with polynomial size, depth $\bigO(\log n)$ and unbounded fan-in.
We show the following theorem.

\begin{theorem}\label{theorem:ac1}
Let $q$ be an $\AC^1$ query over ordered databases and $\Delta$ a finite set of parameter-free first-order definable replacement queries. Then $(q, \Delta)$ is in $\DynFO$ with suitable initialization.
\end{theorem}
\begin{proof}
We first explain the idea underlying the proof. 

It uses the characterization of $\AC^1$ by iterated first-order formulas. More precisely, we use the equality 
$\AC^1=\Ind{\log n}$ from \cite[Theorem 5.22]{ImmermanDC}, where $\Ind{t(n)}$ is the class of problems that can be expressed by applying a first-order formula $\bigO(t(n))$ times and $n$ is the size of the domain\footnote{In the setting of \cite{ImmermanDC}, first-order formulas may use built-in relations $\leq$ and \BIT. The relation $\leq$ is also present here, the relation $\BIT$ can be generated by a suitable initialization, see \cite[Exercise 4.18]{ImmermanDC}.}. We only give an example and refer to \cite[Definition 4.16]{ImmermanDC} for a formal definition. Consider the formula $\varphi_{\text{TC}}(x, y) = (x = y) \vee E(x,y) \vee \exists z \big( R(x, z) \wedge R(z, y) \big)$. When applying the formula to a graph and an empty relation $R$ it defines the relation $R_1$ of paths of length $1$, applying it to $ R \df R_1$ defines the paths of length  $2$; in general applying the formula to $R \df R_i$ defines the paths of length $2^i$. Thus $\log n$-fold application of $\varphi_{\text{TC}}$ defines the transitive closure relation of a graph with $n$ vertices and therefore $\Qreach$ is in $\Ind{\log n}$.

Let $q$ be a query in $\AC^1$ and let $k$ be such that $q$ can be evaluated by $k\log n$ applications of a formula $\varphi_q$.

The program $\prog$ uses a technique inspired from prefetching, which was called \emph{squirrel technique} in \cite{ZeumeS14}. At any point $t$ in time\footnote{We count the occurrence of one change as one time step.}, it starts a thread $\theta_\beta$, for each possible future sequence $\beta$ of $2\log n$ change operations.

Within the next $\log n$ steps (i.e.~changes), it compares whether the actual change sequence $\alpha$ is the prefix of $\beta$ of length $\log n$. If not, thread $\theta_\beta$ is abandoned, as soon as  $\alpha$ departs from~$\beta$. For each of these $\log n$ steps, $\theta_\beta$ simulates two change operations of $\beta$ and applies them to the graph $G_t$ at time $t$, consecutively. After  $\log n$ steps, that is, at time $t+\log n$, thread $\theta_\beta$ has computed the \emph{target graph} $\beta(G_t)$. 

During the next $\log n$ steps until time $t+2\log n$, $\theta_\beta$ evaluates $q$ on 
$\beta(G_t)$ by repeatedly applying the formula $\varphi_q$, $k$ times for each single step. Again, if the actual change sequence departs from $\beta$ then $\theta_\beta$ is abandoned. However, if $\beta$ is the actual change sequence from time $t$ to $t+2\log n$, the thread $\theta_\beta$ does not stop and has the correct query result $q(\beta(G_t))$ at time $t+2\log n$. 

 We note that, although the time window in the above sketch stretches over $2\log n$ change operations from time $t$ to $t+2\log n$, the actual sequences whose effect on the current graph is precomputed are never longer than $\log n$. This is because the application of all $2\log n$ operations of a sequence takes until time $t+\log n$ and by that time the first $\log n$ of these operations already lie in the past. 

Of course, $\prog$ uses a lot of prefetching. However, this is possible, because only a constant number, $d=|\Delta|$,  of change operations is available at any time (and there are no parameters). Thus, there are only $d^{2\log n}=2^{2\log d \log n}=n^{2\log d}$ many different change sequences, each of which can be encoded by a tuple of arity $2\log d$ over the domain.  

This explains how $\prog$ can give correct answers for all times $t\ge 2 \log n$. All previous time points have to be dealt with by the initialization. This initialization also equips the program with the \BIT\ relation. Clearly, the initialization can be computed in $\AC^1$, and therefore also in polynomial time.
More  details of this proof can be found in the full version of this paper.
\end{proof}

  \section{Formal Languages  and $\Sigma_1$-definable Change Operations}\label{section:languages}
In this section, we consider the membership problem for formal languages and how it can be maintained, for regular and context-free languages, under certain kinds of complex changes.

The problem of maintaining formal languages dynamically has been studied intensely in the context of single insertions to and deletions from the relations $R_\sigma$ (cf.\ Section~\ref{section:preliminaries}). 
In that setting, the class of regular languages is exactly the class of languages maintainable in \DynProp\footnote{So, only using quantifier-free update formulas.} and all context-free languages can be maintained in \DynFO \cite{GeladeMS12}. All regular, some context-free, and even some non-context-free languages  can be maintained in \DynFO with only unary auxiliary relations \cite{Hesse03}, but this is not possible for all context-free languages \cite{Zeume15, Vortmeier13}.

Here, we consider the problem of maintaining formal languages under first-order definable change operations. 
We assume that only replacement queries are used whose application results in structures where each position is in at most one $R_\sigma$ relation.
For a given formal language $L$ we denote the membership query for $L$ as $q_L$. 

We prove that regular and context-free languages can be maintained dynamically for large classes of change operations: all regular languages can be maintained in \DynProp under quantifier-free change operations and all context-free languages can be maintained in \DynFO under \EFO-definable (and, dually, \AFO-definable) change operations. A setting, in which language membership can be maintained with respect to simple changes but not with respect to definable change operations is exhibited in Section~\ref{section:lower}.
For quantifier-free change operations, the results are obtained by generalizations of the techniques of \cite{GeladeMS12}.

\begin{theorem}\label{theorem:languages:DynProp}
Let $L$ be a regular language and $\Delta$ a finite set of
quantifier-free replacement queries. Then $(q_L,\Delta)$ can be
maintained in \DynProp  with suitable initialization.
\end{theorem}
\begin{proof}
Let $L$ be a regular language of strings over alphabet $\Sigma$ and $\calA = (Q, \Sigma, \gamma, s, F)$ a corresponding deterministic finite automaton with set $Q$ of states, transition function\footnote{Since in this paper $\delta$ denotes change operations, we use $\gamma$ for transition functions.} $\gamma$, initial state $s$, and set $F$ of accepting states. 
In \cite[Proposition 3.3]{GeladeMS12}, the main auxiliary relations are of the form $S_{q,r}(i,j)$ where $q,r$ are states of $\calA$ and $i,j$ are positions of the string under consideration. The intended meaning of $S_{q,r}$ is that $(i,j) \in S_{q,r}$ if and only if $\gamma^*(q,w_{i+1}\cdots w_{j-1})=r$.\footnote{The relations $S_{q,r}$ were actually named $R_{q,r}$ in \cite{GeladeMS12}, but we want to avoid confusion with the $R_\sigma$ relations. Since \cite{GeladeMS12} did not use constants $\min$ and $\max$, it used further auxiliary relations of the form $I_r$ and $F_q$ that contain all positions $i$ with $\gamma^*(s,w_1\cdots w_{i-1})=r$, and $\gamma^*(q,w_{i+1}\cdots w_n)\in F$, respectively.} Notice that $w_i$ and $w_j$ are not relevant for determining whether $(i,j) \in S_{q,r}$.

In the presence of quantifier-free change operations it suffices to maintain binary auxiliary relations of the form $S^f_{q,r}$, where $f:\Sigmae\to\Sigmae$ is a \emph{relabeling function}. The intended meaning is that  $(i,j) \in S^f_{q,r}$ if and only if 
$\gamma^*(q,f(w_{i+1}\cdots w_{j-1}))=r$, where $f$ is extended to strings in the straightforward way.\footnote{It should be noted that $f$ need not be a homomorphism since $f(\epsilon)\not=\epsilon$ is allowed.} Clearly, $S_{q,r}=S^\id_{q,r}$. 

For simplicity we show how to update $S^f_{q,r}$ for replacement queries of the form $\rho(p)$ with \emph{one} parameter $p$. The general case works analogously, but is notationally more involved. A replacement query with one parameter basically consists of one quantifier free formula $\mu_\sigma(p;x)$, for each element $\sigma\in\Sigma$. 

We show how the relations $S^f_{q,r}$ can be maintained by quantifier-free update formulas $\phi^{f}_{q,r}(p;x,y)$. Then the (Boolean) query relation can be updated by the formula 
\[\bigvee_{q,r\in Q,r'\in F} \psi_{s,q}(\min)\land \phi^{\id}_{q,r}(p;x,y)(\min,\max)\land  \psi_{r,r'}(\max),\] 
where $\psi_{q,r}(x)$ is a formula that expresses that $\delta(q,w_x)=r$.

Intuitively, each formula $\mu_\sigma(p;x)$ determines whether position $x$ carries $\sigma$ after the change. Whether this is the case only depends on (1) the current symbol at position~$x$, (2) the current symbol at position $p$, and (3) on the relative order of $x$ and $p$. Thus, the impact of a change can be described as follows: some relabeling function $f_{\leftarrow}$ is applied at all positions $x<p$, some change might occur at position $p$ and some relabeling function $f_{\rightarrow}$ is applied at all positions $x>p$. More precisely, from $\rho$ one can derive, for each\footnote{Since the schema is clear from the context, we use $\tau$ here to denote a symbol from $\Sigma$.} $\tau\in\Sigmae$, relabeling functions $f^\tau_\leftarrow$, $f^\tau_\rightarrow$ and a symbol $\sigma(\tau)$ such that the update formula for a relation $S^f_{q,r}$ can be described by the formula
\[
\begin{split}
\phi^f_{q,r}(p;x,y)= x < y \wedge \bigvee_{\tau\in\Sigmae}  \bigg(R_\tau(p) \land 
\Big( (p\le x \land  R^{f\circ f^\tau_\rightarrow}_{q,r}(x,y)) \lor   
(p\ge y \land  R^{f\circ f^\tau_\leftarrow}_{q,r}(x,y)) \lor\\ 
\big(x<p<y \land \bigvee_{q',r'}  (R^{f\circ f^\tau_\leftarrow}_{q,q'}(x,p) \land \chi_{q',f(\sigma(\tau)),r'}\land R^{f\circ f^\tau_\rightarrow}_{r',r}(p,y))\big) 
\Big)
\bigg),
\end{split}
\]
where formulas of the form $\chi_{q',a,r'}$ are defined as $\top$ if $\delta(q',a)=r'$ and $\bot$, otherwise.

The initialization of the relations $S^f_{q,r}$ is straightforward. If, for a relabeling function $f$, $f(\epsilon)=\sigma$ then a pair $(i,j)$ is in $S^f_{q,r}$ if and only if $\gamma^*(q,\sigma^{j-i-1})=r$.
\end{proof}

We next turn to context-free languages. The ideas underlying the proof of Theorem~\ref{theorem:languages:DynProp} can be adapted to show that the result from \cite{GeladeMS12}, that (membership for) context-free languages can be maintained in \DynFO under simple change operations, survives under quantifier-free change operations. Through some little extra effort, this can be extended to $\EFO$-definable change operations (and dually, $\AFO$-definable change operations).

\begin{theorem}\label{theorem:languages:EFO}
Let $L$ be a context-free language and $\Delta$ a finite set of
$\EFO$-definable replacement queries. Then $(q_L,\Delta)$ can be
maintained in \DynFO with suitable initialization.
 \end{theorem}
The proof of Theorem \ref{theorem:languages:EFO} can be found in the full version of this article. It first shows how
context-free languages can be maintained under quantifier-free
changes, basically combining the idea of the proof of Theorem~\ref{theorem:languages:DynProp}
with that of \cite[Theorem 4.1]{GeladeMS12}. Then it shows how the
case of \EFO-definable changes can be reduced to the quantifier-free
case.

  \section{Inexpressibility Results}\label{section:lower}
We finally turn to inexpressibility results. It is notoriously difficult to show that a query \emph{cannot} be maintained by a \DynFO program. Indeed, there are no inexpressibility results for \DynFO besides those that follow from the easy observation that every query that can be maintained in \DynFO under single-tuple insertions is computable in polynomial time.

We expect that it should be easier to prove inexpressibility results
for \DynFO in the presence of first-order definable change
operations. However, we have no results of this form yet. But the
following results confirm that, unsurprisingly, complex change
operations can make it harder to maintain a
query. We give two examples where allowing complex changes destroy a
previous maintainability result, Theorems~\ref{theorem:tc:general:insertion} and~\ref{theorem:languages:EAFO}, and one
example, Theorem~\ref{theorem:tc:lower:quantifier-free} where we are able to show an inexpressibility result in the
presence of complex deletions but not yet for single-tuple changes.

Towards our first result, we recall that the reachability query can be
maintained under single-tuple insertions with the transitive closure of the edge
relation as only auxiliary relation and that this does not hold if one
allows single-tuple deletions \cite{DongLW95}. We show next that the
transitive closure also does not suffice in the presence of
single-tuple insertions and one complex insertion query.

For general directed graphs, a parameter-free and quantifier-free
insertion query suffices, for acyclic graphs a parameter-free insertion query defined by an existential formula suffices. The latter result should be contrasted with Theorem \ref{theorem:tc:acyclic:quantifierfree}.

\begin{theorem}\label{theorem:tc:general:insertion}
\begin{enumerate}[(a)]
\item There is a quantifier-free and parameter-free insertion query $\rho$ such that
$(\Qreach,\{\insertion{E}{\tpl p},\rho\})$ cannot be maintained in \DynFO
on ordered directed graphs, if all auxiliary relations besides the query
relation and the linear order are unary.
\item There exists an \EFO-definable and parameter-free insertion query $\rho'$, for which the above statement holds even restricted to \emph{acyclic}, directed graphs.
\end{enumerate}
\end{theorem} 
\begin{proof}
For ease of presentation, we first give a proof for  unordered directed graphs.

The proof follows an approach that has been used often before and that was made precise in \cite{Zeume15}. We say that a $k$-ary query $q$ is expressed by a formula $\varphi(\tpl x)$ with \emph{help relations of schema $\tau$}, if, for every database $\db$, there is a $\tau$-structure $H$ over the same domain such that for every $k$-tuple $\tpl a$ over the domain of $\db$ it holds\footnote{This notion should not be confused with definability of the query $q$ in existential second-order logic. In the latter case, the relations can be chosen depending on $\tpl a$, but here the relations need to ``work'' for all tuples $\tpl a$.}: $\tpl a \in q(\db)$ if and only if $(\db,H)\models \varphi(\tpl a)$.

The proof is by contradiction and proceeds in the same way in both
cases, (a) and~\mbox{(b)}.  Our goal is to show that, under the assumption that there is a dynamic program for (a) or~(b), the transitive closure of path graphs, that is, graphs that consist of a single directed path, can be expressed with unary help relations, contradicting the following lemma from~\cite{Zeume15}, which is not hard to prove with the help of locality arguments.

\begin{lemma}[{\cite[Lemma 4.3.2]{Zeume15}}]\label{lemma:tc:pathgraphs}
The transitive closure of path graphs cannot be expressed by
a first-order formula with unary help relations.
\end{lemma}

We refer to
Figure~\ref{figure:tc:general:insertiona} for an illustration of the
following high-level sketch. We start from an arbitrary path graph $G_0$ and equip it with some unary
relations $C_0,C_1,C_2$. From $G_0,C_0,C_1,C_2$ we define a graph $G$ in a
first-order fashion, whose simple directed paths have length at most
2, so the transitive closure relation $\TC$ of $G$ is definable by a first-order formula. 
Finally, the crucial step happens: the change operation $\delta$ 
transforms $G$ into a graph $G'=\delta(G)$ with the property
that $\Qreach(G_0)$ can be defined from $\Qreach(G')$ by a first-order formula.
We can conclude that $\Qreach(G_0)$ can be defined by a first-order formula with the help
of suitable unary help relations,  since all steps from $G_0$ to $G'$ are
first-order definable, $\TC$ is first-order definable from $G$, and we
assume that there is a dynamic program that computes $\Qreach(G')$
from $G$, $\TC$ and some unary auxiliary relations. This contradicts Lemma \ref{lemma:tc:pathgraphs}.

For (a), we use the insertion query $\rho = \mf{E}{x,y}{} \df E(x,y) \lor
\big( E(x,x) \land E(y,y) \land E(y,x) \big)$ that adds all edges
$(x,y)$ for which there is an edge $(y,x)$ and both $x$ and $y$ have
self-loops.  We assume that there is a \DynFO-program $\prog$ that
maintains the reachability query on directed graphs under
insertion queries $\{\insertion{E}{\tpl p}, \rho\}$. We further assume that $\prog$ uses (only) unary auxiliary relations $B_1,\ldots,B_k$, for some $k$, besides the binary relation $Q$ intended to store the query result. We show how to construct from $\prog$ a first-order formula $\varphi$ that expresses the reachability query $\Qreach$ for simple paths with unary help relations $B_1,\ldots,B_k, C_0,C_1,C_2$, contradicting Lemma \ref{lemma:tc:pathgraphs}.

Let $G_0=(V_0,E_0)$ be a simple path with $V_0=\{v_0,\ldots,v_n\}$, for which we
want to define $\Qreach$ using unary help relations $B_1,\ldots, B_k, C_0, C_1, C_2$. Let $C_0,C_1,C_2$ be defined by $C_i=\{v_j\mid 0\le j\le n, j \equiv_3 i\}$ where $\equiv_3$ denotes modulo 3 equivalence.
From $G_0$ and $C_0, C_1, C_2$ we define the following graph $G$ with nodes $v_0,\ldots,v_n$.
The graph $G$ has an edge from vertex $v$ to $w$ if one of the following cases holds:
\begin{itemize}
\item $v \in C_0, w \in C_1$ and $(v,w)$ is an edge in $G_0$, 
\item  $v \in C_1, w \in C_2$ and $(w,v)$ is an edge in $G_0$, 
\item $v \in C_2, w \in C_0$ and $(v,w)$ is an edge in $G_0$, or
\item $v \in C_1\cup C_2$ and $v=w$.
\end{itemize}
We observe that the graph $G$ can be first-order defined from $G_0$ and $C_0, C_1, C_2$.

  \begin{figure}[ht] 
\begin{tikzpicture}[
      xscale=1.0,
      yscale=1.0,
      font=\footnotesize,
    ]
          \node (tmp) at (-1,0) {$G_0$:};
          \node (0) at (0,0)[mnode, label=below:$v_0$, label=above:$C_0$] {};
          \node (1) at (1,0)[mnode, label=below:$v_1$, label=above:$C_1$] {};
          \node (2) at (2,0)[mnode, label=below:$v_2$, label=above:$C_2$] {};
          \node (3) at (3,0)[mnode, label=below:$v_3$, label=above:$C_0$] {};
          \node (4) at (4,0)[mnode, label=below:$v_4$, label=above:$C_1$] {};
          \node (5) at (5,0)[mnode, label=below:$v_5$, label=above:$C_2$] {};
          \node (6) at (6,0)[mnode, label=below:$v_6$, label=above:$C_0$] {};
          \node (7) at (7,0)[mnode, label=below:$v_7$, label=above:$C_1$] {};
          \node (8) at (8,0)[mnode, label=below:$v_8$, label=above:$C_2$] {};
          \node (9) at (9,0)[mnode, label=below:$v_9$, label=above:$C_0$] {};
          \node (10) at (10,0)[mnode, label=below:$v_{10}$, label=above:$C_1$] {};
          \node (11) at (11,0)[mnode, label=below:$v_{11}$, label=above:$C_2$] {};
 
           \draw [dEdge] (0) to (1);
           \draw [dEdge] (1) to (2);
           \draw [dEdge] (2) to (3);
           \draw [dEdge] (3) to (4);
           \draw [dEdge] (4) to (5);
           \draw [dEdge] (5) to (6);
           \draw [dEdge] (6) to (7);
           \draw [dEdge] (7) to (8);
           \draw [dEdge] (8) to (9);
           \draw [dEdge] (9) to (10);
           \draw [dEdge] (10) to (11);

    \end{tikzpicture}  
  
    \begin{tikzpicture}[
      xscale=1.0,
      yscale=1.0,
      font=\footnotesize,
    ]
          \node (tmp) at (-1,0) {$G$:};
          \node (0) at (0,0)[mnode, label=below:$v_0$] {};
          \node (1) at (1,0)[mnode, label=below:$v_1$] {};
          \node (2) at (2,0)[mnode, label=below:$v_2$] {};
          \node (3) at (3,0)[mnode, label=below:$v_3$] {};
          \node (4) at (4,0)[mnode, label=below:$v_4$] {};
          \node (5) at (5,0)[mnode, label=below:$v_5$] {};
          \node (6) at (6,0)[mnode, label=below:$v_6$] {};
          \node (7) at (7,0)[mnode, label=below:$v_7$] {};
          \node (8) at (8,0)[mnode, label=below:$v_8$] {};
          \node (9) at (9,0)[mnode, label=below:$v_9$] {};
          \node (10) at (10,0)[mnode, label=below:$v_{10}$] {};
          \node (11) at (11,0)[mnode, label=below:$v_{11}$] {};
 
           \draw [dEdge] (0) to (1);
           \draw [dEdge] (2) to (1);
           \draw [dEdge] (2) to (3);
           \draw [dEdge] (3) to (4);
           \draw [dEdge] (5) to (4);
           \draw [dEdge] (5) to (6);
           \draw [dEdge] (6) to (7);
           \draw [dEdge] (8) to (7);
           \draw [dEdge] (8) to (9);
           \draw [dEdge] (9) to (10);
           \draw [dEdge] (11) to (10);

           \path[->] (1) edge  [dEdge,loop above] node {} ();
           \path[->] (2) edge  [dEdge,loop above] node {} ();
           \path[->] (4) edge  [dEdge,loop above] node {} ();
           \path[->] (5) edge  [dEdge,loop above] node {} ();
           \path[->] (7) edge  [dEdge,loop above] node {} ();
           \path[->] (8) edge  [dEdge,loop above] node {} ();
           \path[->] (10) edge  [dEdge,loop above] node {} ();
           \path[->] (11) edge  [dEdge,loop above] node {} ();
                      
    \end{tikzpicture}

        \begin{tikzpicture}[
      xscale=1.0,
      yscale=1.0,
      font=\footnotesize
    ]
          \node (tmp) at (-1,0) {$G'$:};
          \node (0) at (0,0)[mnode, label=below:$v_0$] {};
          \node (1) at (1,0)[mnode, label=below:$v_1$] {};
          \node (2) at (2,0)[mnode, label=below:$v_2$] {};
          \node (3) at (3,0)[mnode, label=below:$v_3$] {};
          \node (4) at (4,0)[mnode, label=below:$v_4$] {};
          \node (5) at (5,0)[mnode, label=below:$v_5$] {};
          \node (6) at (6,0)[mnode, label=below:$v_6$] {};
          \node (7) at (7,0)[mnode, label=below:$v_7$] {};
          \node (8) at (8,0)[mnode, label=below:$v_8$] {};
          \node (9) at (9,0)[mnode, label=below:$v_9$] {};
          \node (10) at (10,0)[mnode, label=below:$v_{10}$] {};
          \node (11) at (11,0)[mnode, label=below:$v_{11}$] {};
 
           \draw [dEdge] (0) to (1);
           \draw [dEdge, bend left] (1) to (2);
           \draw [dEdge, bend left] (2) to (1);
           \draw [dEdge] (2) to (3);
           \draw [dEdge] (3) to (4);
           \draw [dEdge, bend left] (4) to (5);
           \draw [dEdge, bend left] (5) to (4);
           \draw [dEdge] (5) to (6);
           \draw [dEdge] (6) to (7);
           \draw [dEdge, bend left] (7) to (8);
           \draw [dEdge, bend left] (8) to (7);
           \draw [dEdge] (8) to (9);
           \draw [dEdge] (9) to (10);
           \draw [dEdge, bend left] (10) to (11);
           \draw [dEdge, bend left] (11) to (10);

           \path[->] (1) edge  [dEdge,loop above] node {} ();
           \path[->] (2) edge  [dEdge,loop above] node {} ();
           \path[->] (4) edge  [dEdge,loop above] node {} ();
           \path[->] (5) edge  [dEdge,loop above] node {} ();
           \path[->] (7) edge  [dEdge,loop above] node {} ();
           \path[->] (8) edge  [dEdge,loop above] node {} ();
           \path[->] (10) edge  [dEdge,loop above] node {} ();
           \path[->] (11) edge  [dEdge,loop above] node {} ();
    \end{tikzpicture}

    \caption{The graphs from the proof of Theorem \ref{theorem:tc:general:insertion}(a).}
      \label{figure:tc:general:insertiona}
      
  \end{figure}
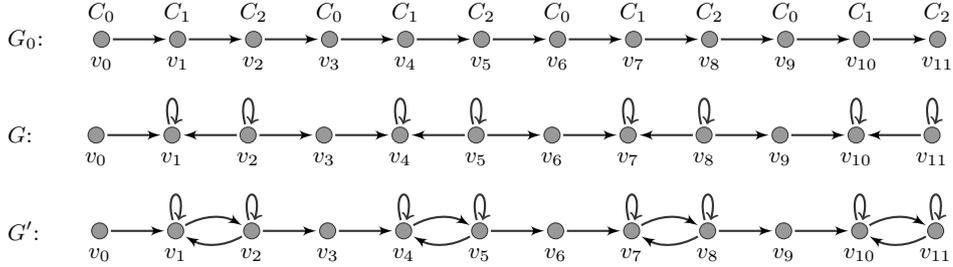

 Let $\delta\df\rho$ and\footnote{Since $\rho$ is parameter-free,
   the insertion query $\rho$  and its corresponding change
 operation $\delta$ are basically the same.} $G'\df\delta(G)$. The graphs $G_0$, $G$ and $G'$ for $n = 11$ are depicted
in Figure~\ref{figure:tc:general:insertiona}. By our assumption, the
update formula $\psi=\uf{Q}{\delta}{x_1,x_2}{}$ of $\prog$ for the
query relation $Q$ and operation $\delta$ defines the reachability
query for $\delta(G)=G'$ with the help of suitable auxiliary relations
$B_1,\ldots,B_k$ and the transitive closure $\TC$ of the edge relation of~$G$. 

Altogether, $G=f(G_0,C_0,C_1,C_2)$, for some first-order definable
function $f$, $\TC$ is first-order definable from $G$, $G'=\delta(G)$, $\Qreach(G')$ is first-order definable
from $G$, $\TC$ and $B_1,\ldots,B_k$, and therefore
\[\Qreach(G_0) = \Qreach(G') \setminus \{(w,v) \mid v \in C_1, w \in
C_2, (v,w) \text{ is an edge in } G_0\}\] is first-order definable
from $G_0$, $C_0,C_1,C_2$, and $B_1,\ldots,B_k$, contradicting Lemma \ref{lemma:tc:pathgraphs}, as desired. 

The proof for (b) and the extension to ordered graphs can be found in the full version of this paper.
\end{proof}

We now turn towards inexpressibility by quantifier-free update formulas. Very likely quantifier-free update formulas are too weak to maintain
$\Qreach$ even under single-tuple changes. Yet only restricted
inexpressibility results have been obtained so far. The query
$\Qreach$ cannot be maintained in \DynProp under single-tuple changes when the auxiliary relations are at most binary or when the initialization is severely restricted \cite{ZeumeS15}. For the more general alternating reachability query quantifier-free update formulas do not suffice~\cite{GeladeMS12}.
The next result shows that $\Qreach$ cannot be maintained in \DynProp,
even if besides
single-edge insertions only a single, very simple deletion query is allowed.

\begin{theorem}\label{theorem:tc:lower:quantifier-free}
  There is a quantifier-free deletion query $\rho$ with one parameter such that $(\Qreach,\Delta_{E}\cup\{\rho\})$ cannot be maintained in \DynProp.
\end{theorem}
\begin{proof}
For the proof, we combine the standard tool for obtaining
inexpressibility results for $\DynProp$, the Substructure Lemma
\cite{ZeumeS15, GeladeMS12}, with a combinatorial technique based on
upper and lower bounds of Ramsey numbers \cite{Zeume14}. 

The intuition behind the Substructure Lemma is as follows. When
updating an auxiliary tuple $\tpl d$ after a quantifier-free change
parameterized by  $\tpl p$, a quantifier-free update formula only has
access to $\tpl d$ and $\tpl p$. Thus, if a change operation  changes
a tuple \emph{inside} a substructure  $\calA$ of a state $\calS$ of a
dynamic program, the auxiliary data of $\calA$ is not affected by any
information from \emph{outside} of $\calA$. In particular, two
isomorphic substructures $\calA$ and $\calB$ remain isomorphic, when
corresponding changes are applied to them. The Substructure Lemma is
formally stated in \cite[Lemma 2]{ZeumeS15}. Even though the lemma is phrased for single-tuple changes only,
the same proof, using the intuition explained above, extends to
quantifier-free replacement queries.

For the actual proof, we assume, towards a contradiction, that there is a quantifier-free
dynamic program $\prog$ over schema $\tau$ of arity $k$ that maintains
$\Qreach$ under the quantifier-free deletion $\rho(p) = \mu(p;
x,y) \df E(x,y) \wedge \neg E(p, x)$ 
which deletes an edge $(x,y)$ if there is an edge $(p,x)$. Our goal is
to construct a graph $G$ such that not all change sequences of length
$k+1$ can be maintained, no matter the initial
auxiliary data.  
  
Let $n$ be a sufficiently large number, to be specified later. The vertex set of the graph is of the form $\{s,t\}\cup A \cup C$, for some disjoint sets $A$ and~$C$, with $|C| = n$.  
The set $A$ contains a node for every subset of size $k+1$ of $C$, that is, $A \df \{a_X \mid X \subseteq C \text{ and } |X| = k+1\}$. 
Let~$B$ be a subset of $A$, to be specified later.

  The graph has the following edges:
  \begin{enumerate}
   \item For each $a_X \in A$ there is an edge $(s, a_X)$. 
   \item For each $a_X \in B$ there is an edge $(a_X, t)$.  
   \item There is an edge $(c, a_X)$ for nodes $c \in C, a_X \in A$ if $c \not\in X$.
  \end{enumerate}

  Intuitively, the nodes in $C$ control how edges from $A$ to $t$ can be removed. Each node $c \in C$ is connected to a subset $A' \subseteq A$, and thus applying a change $\rho(c)$ will result in removing all edges $(a_X, t)$ for all $a_X \in  A'$. The graph is constructed in such a way that
  \begin{enumerate}
    \item[($\star$)] for a change sequence $\alpha=(\rho(c_1),\ldots, \rho(c_{k+1}))$ with $|\{c_1, \ldots, c_{k+1}\}| = k+1$ it holds $(s,t)\in \Qreach(\alpha(G))$ if and only if $a_{\{c_1,\ldots,c_{k+1}\}} \in B$. 
  \end{enumerate}

  To see this, observe that after applying changes $\rho(c_1), \ldots, \rho(c_{k+1})$, all edges $(a_X,t)$ are deleted, for which $\{c_1,\ldots,c_{k+1}\}\not\subseteq X$.
  Thus at most the edge $(a_{\{c_1,\ldots,c_{k+1}\}}, t)$ is still present. However, this edge was at all present in the graph if and only if $a_{\{c_1,\ldots,c_{k+1}\}} \in B$.

  For choosing the size of $C$ and the set $B$, we employ the
  combinatorial Lemma 2 from~\cite{Zeume14}. The lemma guarantees
  that, depending on the schema $\tau \cup \{c_s, c_t\}$, there is an $n_0$ such that
  for every $n>n_0$ there is some $m$ such that the following holds.
  \begin{description}
  \item[(S1)] For every state $\state$ of the dynamic program for $G$,
    and each set $C$ with at least $n$ vertices of $G$ with a
    linear order $<$, there is a subset $C'$ of $C$ of size at least
    $m$ such that the $k$-ary auxiliary data on $C'$ is
    <-monochromatic in the structure $(\state, s, t)$, i.e. all
    $<$-ordered $k$-tuples over $C'$ have the same quantifier-free
    type (including their relationships to the interpretations $s, t$
    of the constants $c_s, c_t$). 
  \item[(S2)] There is a subset
    $B$ of $A$ such that for every subset $\hat{C}$ of $C$ of size
    $m$, there are $(k+1)$-element sets
    $Y = \{c_1, \ldots, c_{k+1}\}, Y' = \{c'_1, \ldots, c'_{k+1}\}
    \subseteq \hat{C}$ with $a_Y \in B$ and $a_{Y'} \notin B$.
  \end{description}

  We outline how the graph $G$ is used to obtain a contradiction. Let $\state$ be a state of the dynamic program for the graph $G$ with $|C| = n > n_0$ and let $<$ be a linear order. Choose $B$ as described above and a subset $C'$ of $C$ of size $|C'| =m$ that is <-monochromatic in $(\state, s, t)$.  Choose $Y = \{c_1, \ldots, c_{k+1}\}, Y' = \{c'_1, \ldots, c'_{k+1}\} \subseteq C'$  with $a_Y \in B$ and $a_{Y'} \notin B$. By the Substructure Lemma from \cite{ZeumeS15} generalized to quantifier-free changes, the dynamic program $\prog$ yields the same result for the tuple $(s,t)$ for the change sequences $(\rho(c_1),\ldots, \rho(c_{k+1}))$ and $(\rho(c'_1),\ldots, \rho(c'_{k+1}))$ since $C'$ is <-monochromatic.  Yet the result should be different due to~($\star$) and $a_Y \in B$, $a_{Y'} \notin B$. This is a contradiction.
\end{proof}

Finally, we turn to lower bounds for the maintenance of languages. We
exhibit an example that illustrates that maintaining regular languages
under full first-order replacement queries might be hard: there
is a regular language $L$ that  can 
be maintained in \DynFO under single-tuple changes with nullary auxiliary relations, but there is a
relatively simple replacement query, for which this no longer holds.
This is no general hardness result, as we only allow very restricted
auxiliary relations, but it demonstrates the barrier of our
techniques. The proof of the following result can be found in the full version of this paper.

\begin{theorem}\label{theorem:languages:EAFO}
There is a regular language $L$ over some alphabet $\Sigma$ and a replacement query $\rho$, such that $(q_L,\Delta_\Sigma)$  can 
be maintained in \DynFO with nullary auxiliary relations, but not
$(q_L,\Delta_\Sigma\cup\{\rho\})$.
\end{theorem}

\section{Conclusion}\label{section:conclusion}
In this paper, we studied the maintainability of queries in the Dynamic Complexity setting under first-order defined replacement queries. The main insight of this study is that many maintainability results carry over  from the single-tuple world to settings with more general change operations. We were actually quite surprised to see that so many positive results survive this transition. However, many questions remain open, for instance: To which extent can the reachability query for (undirected or acyclic) graphs be maintained under definable deletions? What about reachability for unrestricted directed graphs under definable insertions? What about other queries? Are binary auxiliary relations sufficient in Theorem~\ref{theorem:tc:undirected:insertion}?

We were less surprised by the fact that stronger change operations can yield inexpressibility, but even these results required some care. Our main contribution in that respect is the proof that \DynProp cannot maintain the reachability query under quantifier-free replacement queries.

From Theorem~\ref{theorem:ac1}  about parameter-free changes and its proof, we take another insight regarding inexpressibility proofs: the squirrel technique is quite powerful to prepare an update program for a non-constant (i.e., logarithmic) number of changes. Inexpressibility proofs need to take that into account and to argue ``around it''.

 \bibliography{schwentick}

\begin{thebibliography}{10}

\bibitem{AmelootBW13}
Tom~J. Ameloot, Jan~Van den Bussche, and Emmanuel Waller.
\newblock On the expressive power of update primitives.
\newblock In {\em Proceedings of the 32nd {ACM} {SIGMOD-SIGACT-SIGART}
  Symposium on Principles of Database Systems ({PODS})}, pages 139--150, 2013.
\newblock \href {http://dx.doi.org/10.1145/2463664.2465218}
  {\path{doi:10.1145/2463664.2465218}}.

\bibitem{DattaHK14}
Samir Datta, William Hesse, and Raghav Kulkarni.
\newblock Dynamic complexity of directed reachability and other problems.
\newblock In {\em Automata, Languages, and Programming - 41st International
  Colloquium ({ICALP}), Proceedings, Part {I}}, pages 356--367, 2014.
\newblock \href {http://dx.doi.org/10.1007/978-3-662-43948-7_30}
  {\path{doi:10.1007/978-3-662-43948-7_30}}.

\bibitem{DattaKMSZ15}
Samir Datta, Raghav Kulkarni, Anish Mukherjee, Thomas Schwentick, and Thomas
  Zeume.
\newblock Reachability is in {DynFO}.
\newblock In {\em Automata, Languages, and Programming - 42nd International
  Colloquium ({ICALP}), Proceedings, Part {II}}, pages 159--170, 2015.
\newblock \href {http://dx.doi.org/10.1007/978-3-662-47666-6_13}
  {\path{doi:10.1007/978-3-662-47666-6_13}}.

\bibitem{DongLW95}
Guozhu Dong, Leonid Libkin, and Limsoon Wong.
\newblock On impossibility of decremental recomputation of recursive queries in
  relational calculus and {SQL}.
\newblock In {\em Proceedings of the Fifth International Workshop on Database
  Programming Languages (DBPL-5)}, page~7, 1995.

\bibitem{DongP97}
Guozhu Dong and Chaoyi Pang.
\newblock Maintaining transitive closure in first order after node-set and
  edge-set deletions.
\newblock {\em Inf. Process. Lett.}, 62(4):193--199, 1997.
\newblock \href {http://dx.doi.org/10.1016/S0020-0190(97)00066-5}
  {\path{doi:10.1016/S0020-0190(97)00066-5}}.

\bibitem{DongS93}
Guozhu Dong and Jianwen Su.
\newblock First-order incremental evaluation of datalog queries.
\newblock In {\em Proceedings of the Fourth International Workshop on Database
  Programming Languages - Object Models and Languages (DBPL-4)}, pages
  295--308, 1993.

\bibitem{DongS95}
Guozhu Dong and Jianwen Su.
\newblock Incremental and decremental evaluation of transitive closure by
  first-order queries.
\newblock {\em Inf. Comput.}, 120(1):101--106, 1995.
\newblock \href {http://dx.doi.org/10.1006/inco.1995.1102}
  {\path{doi:10.1006/inco.1995.1102}}.

\bibitem{DongST95}
Guozhu Dong, Jianwen Su, and Rodney~W. Topor.
\newblock Nonrecursive incremental evaluation of datalog queries.
\newblock {\em Ann. Math. Artif. Intell.}, 14(2-4):187--223, 1995.
\newblock \href {http://dx.doi.org/10.1007/BF01530820}
  {\path{doi:10.1007/BF01530820}}.

\bibitem{DongT92}
Guozhu Dong and Rodney~W. Topor.
\newblock Incremental evaluation of datalog queries.
\newblock In {\em Proceedings of the 4th International Conference on Database
  Theory (ICDT)}, pages 282--296, 1992.
\newblock \href {http://dx.doi.org/10.1007/3-540-56039-4_48}
  {\path{doi:10.1007/3-540-56039-4_48}}.

\bibitem{Etessami98}
Kousha Etessami.
\newblock Dynamic tree isomorphism via first-order updates.
\newblock In {\em Proceedings of the Seventeenth {ACM} {SIGACT-SIGMOD-SIGART}
  Symposium on Principles of Database Systems (PODS)}, pages 235--243, 1998.
\newblock \href {http://dx.doi.org/10.1145/275487.275514}
  {\path{doi:10.1145/275487.275514}}.

\bibitem{Feferman57}
Solomon Feferman.
\newblock Some recent work of {E}hrenfeucht and {F}ra\"\i ss\'e.
\newblock In {\em Proc. Summer Institute of Symbolic Logic}, pages 201--209,
  1957.

\bibitem{FefermanV59}
Solomon Feferman and Robert~L. Vaught.
\newblock The first order properties of algebraic systems.
\newblock {\em Fund. Math.}, 47:57--103, 1959.

\bibitem{GeladeMS12}
Wouter Gelade, Marcel Marquardt, and Thomas Schwentick.
\newblock The dynamic complexity of formal languages.
\newblock {\em {ACM} Trans. Comput. Log.}, 13(3):19, 2012.
\newblock \href {http://dx.doi.org/10.1145/2287718.2287719}
  {\path{doi:10.1145/2287718.2287719}}.

\bibitem{GradelS12}
Erich Gr{\"{a}}del and Sebastian Siebertz.
\newblock Dynamic definability.
\newblock In {\em 15th International Conference on Database Theory ({ICDT})},
  pages 236--248, 2012.
\newblock \href {http://dx.doi.org/10.1145/2274576.2274601}
  {\path{doi:10.1145/2274576.2274601}}.

\bibitem{GuptaMS93}
Ashish Gupta, Inderpal~Singh Mumick, and V.~S. Subrahmanian.
\newblock Maintaining views incrementally.
\newblock In {\em Proceedings of the 1993 {ACM} {SIGMOD} International
  Conference on Management of Data}, pages 157--166, 1993.
\newblock \href {http://dx.doi.org/10.1145/170035.170066}
  {\path{doi:10.1145/170035.170066}}.

\bibitem{Hesse03}
William Hesse.
\newblock {\em Dynamic Computational Complexity}.
\newblock PhD thesis, University of Massachusetts Amherst, 2003.

\bibitem{HesseI02}
William Hesse and Neil Immerman.
\newblock Complete problems for dynamic complexity classes.
\newblock In {\em 17th {IEEE} Symposium on Logic in Computer Science (LICS),
  Proceedings}, page 313, 2002.
\newblock \href {http://dx.doi.org/10.1109/LICS.2002.1029839}
  {\path{doi:10.1109/LICS.2002.1029839}}.

\bibitem{ImmermanDC}
Neil Immerman.
\newblock {\em Descriptive complexity}.
\newblock Graduate texts in computer science. Springer, 1999.
\newblock \href {http://dx.doi.org/10.1007/978-1-4612-0539-5}
  {\path{doi:10.1007/978-1-4612-0539-5}}.

\bibitem{Koch10}
Christoph Koch.
\newblock Incremental query evaluation in a ring of databases.
\newblock In {\em Proceedings of the Twenty-Ninth {ACM} {SIGMOD-SIGACT-SIGART}
  Symposium on Principles of Database Systems (PODS)}, pages 87--98, 2010.
\newblock \href {http://dx.doi.org/10.1145/1807085.1807100}
  {\path{doi:10.1145/1807085.1807100}}.

\bibitem{Libkin04}
Leonid Libkin.
\newblock {\em Elements of Finite Model Theory}.
\newblock Springer, 2004.

\bibitem{Makowsky04}
Johann~A. Makowsky.
\newblock Algorithmic uses of the {F}eferman-{V}aught {T}heorem.
\newblock {\em Ann. Pure Appl. Logic}, 126(1-3):159--213, 2004.
\newblock \href {http://dx.doi.org/10.1016/j.apal.2003.11.002}
  {\path{doi:10.1016/j.apal.2003.11.002}}.

\bibitem{PangDR05}
Chaoyi Pang, Guozhu Dong, and Kotagiri Ramamohanarao.
\newblock Incremental maintenance of shortest distance and transitive closure
  in first-order logic and {SQL}.
\newblock {\em {ACM} Trans. Database Syst.}, 30(3):698--721, 2005.
\newblock \href {http://dx.doi.org/10.1145/1093382.1093384}
  {\path{doi:10.1145/1093382.1093384}}.

\bibitem{PatnaikI97}
Sushant Patnaik and Neil Immerman.
\newblock Dyn-{FO}: {A} parallel, dynamic complexity class.
\newblock {\em J. Comput. Syst. Sci.}, 55(2):199--209, 1997.
\newblock \href {http://dx.doi.org/10.1006/jcss.1997.1520}
  {\path{doi:10.1006/jcss.1997.1520}}.

\bibitem{Schwentick96}
Thomas Schwentick.
\newblock On winning {E}hrenfeucht games and monadic {NP}.
\newblock {\em Ann. Pure Appl. Logic}, 79(1):61--92, 1996.
\newblock \href {http://dx.doi.org/10.1016/0168-0072(95)00030-5}
  {\path{doi:10.1016/0168-0072(95)00030-5}}.

\bibitem{SchwentickZ16}
Thomas Schwentick and Thomas Zeume.
\newblock Dynamic complexity: recent updates.
\newblock {\em {SIGLOG} News}, 3(2):30--52, 2016.
\newblock \href {http://dx.doi.org/10.1145/2948896.2948899}
  {\path{doi:10.1145/2948896.2948899}}.

\bibitem{Siebertz11}
Sebastian Siebertz.
\newblock Dynamic definability.
\newblock Diploma thesis, RWTH Aachen, 2011.

\bibitem{Vortmeier13}
Nils Vortmeier.
\newblock {Komplexit{\"a}tstheorie verlaufsunabh{\"a}ngiger dynamischer
  Programme}.
\newblock Master's thesis, TU Dortmund, 2013.

\bibitem{WeberS07}
Volker Weber and Thomas Schwentick.
\newblock Dynamic complexity theory revisited.
\newblock {\em Theory Comput. Syst.}, 40(4):355--377, 2007.
\newblock \href {http://dx.doi.org/10.1007/s00224-006-1312-0}
  {\path{doi:10.1007/s00224-006-1312-0}}.

\bibitem{Zeume14}
Thomas Zeume.
\newblock The dynamic descriptive complexity of k-clique.
\newblock In {\em Mathematical Foundations of Computer Science (MFCS) - 39th
  International Symposium, Proceedings, Part {I}}, pages 547--558, 2014.
\newblock \href {http://dx.doi.org/10.1007/978-3-662-44522-8_46}
  {\path{doi:10.1007/978-3-662-44522-8_46}}.

\bibitem{Zeume15}
Thomas Zeume.
\newblock {\em Small Dynamic Complexity Classes}.
\newblock PhD thesis, TU Dortmund University, 2015.

\bibitem{ZeumeS14}
Thomas Zeume and Thomas Schwentick.
\newblock Dynamic conjunctive queries.
\newblock In {\em Proc. 17th International Conference on Database Theory
  (ICDT)}, pages 38--49, 2014.
\newblock \href {http://dx.doi.org/10.5441/002/icdt.2014.08}
  {\path{doi:10.5441/002/icdt.2014.08}}.

\bibitem{ZeumeS15}
Thomas Zeume and Thomas Schwentick.
\newblock On the quantifier-free dynamic complexity of reachability.
\newblock {\em Inf. Comput.}, 240:108--129, 2015.
\newblock \href {http://dx.doi.org/10.1016/j.ic.2014.09.011}
  {\path{doi:10.1016/j.ic.2014.09.011}}.

\end{thebibliography}

\section*{Appendix}\label{section:appendix}
\begin{proofof}{Lemma~\ref{lemma:tc:acyclic:boundnewedges}}
Let $\ell$ be the arity of the parameter tuple of $\rho$ and $m'$ the number of $\FO[0,\ell+1]$-types of graphs.  We choose $m=m'+1$. Let further $G$ be a directed, acyclic graph, and let $u,v$ be nodes of $G$.

As in the proof Lemma~\ref{lemma:tc:undirected:boundnewedges}, we show that each path $\pi$ from $u$ to $v$ with $q>m'$ bridges can be transformed into a path with less than $q$ bridges, unless a cycle with at most $m$ bridges is introduced.

To this end, let  $(u_1,v_1),\ldots,(u_q,v_q)$ be the bridges in $\pi$. Since $q$ is larger than the number $m'$ of $\FO[0,\ell+1]$-types, there are $i,j$ with $1 \leq i < j \leq m'$ such that $(G,v_i,\tpl a)$ has the same $\FO[0]$-type as $(G,v_j,\tpl a)$.
We distinguish three cases. In case (1), the edge $(v_i,u_i)$ is in $G$ and thus $\delta$ introduces a cycle of length 2. In case (2), the edge $(v_j,u_i)$ is in $G$ and, together with the sub-path from $u_i$ to $v_j$, constitutes a cycle with at most $m$ bridges. In case (3), $u_i$ is neither connected to $v_i$ nor to $v_j$ by an edge. Therefore  $(u_i,v_i,\tpl a)$ and $(u_i,v_j,\tpl a)$ have the same $\FO[0]$-type and $\delta$ inserts an edge $(u_i,v_j)$ as well, the desired shortcut.
\end{proofof}

\begin{proofof}{Theorem~\ref{theorem:ac1}, continued}
  For simplicity, 
we present a more detailed proof for ordered graphs ($\inpSchema = \{E, \leq\}$), Boolean queries, $d=2$ and $k=1$. That is, there are only
two change operations, $\rho_0$ and $\rho_1$, and $q(G)$ can be obtained by $\log n$ applications of $\varphi_q$. The proof can easily be generalized to the case of general structures as well as arbitrary $d$ and $k$.

Let $G$ be an ordered graph with $n$ vertices. 
For simplicity we assume that $n$ is a power of~$2$.

We encode change sequences by elements of the domain as follows:  A sequence $\alpha=\delta_1\cdots \delta_{\log n}$ is encoded by the node $w_\alpha$, whose bit string representation (when viewed as a number) has 1 at position $i$ if and only if $\delta_i=\rho_1$.  We denote the change sequence encoded by node $w$ as $\alpha_w$.

We first describe the auxiliary relations needed for time points $t\ge \log n$. 
We denote the length $i$ prefix of a change sequence $\alpha$ by $\alpha[..i]$. 
In the following, we assume that relations for the arithmetic operations $+$, $\times$ and the $\BIT$-predicate are provided by the initialization\footnote{In fact, the $\BIT$-predicate is sufficient, as $+$ and $\times$ are \FO-definable from $\BIT$ \cite{ImmermanDC}.}. 

\begin{itemize}
 \item Relation $F$ represents the graphs that result when  particular change sequences occur during the next $\log n$ steps.  The tuple $(w,u,v)$ is in $F$ if after applying $\alpha_w$ to the current graph $G$, the edge $(u,v)$ is included in the resulting graph $\alpha_w(G)$. 
 \item Relation $T$ contains the temporary query result after some applications of $\varphi_q$ to a modified graph: The tuple $(i,w,\tpl a)$ is in $T$ if after $i+1$ applications of $\varphi_q$ to $\alpha_w[..j](G)$, the tuple $\tpl a$ is included in the defined relation, for $i \leq \log n$ and $j=(\log n)-i-1$. Here, $j$ is the length of the prefix of $\alpha_w$ that still needs to be applied to the current graph $G$ to obtain the target graph. In particular, a tuple $\tpl a$ is in the query result for the current graph if $(\log n,w,\tpl a) \in T$ for all $w$.
 \item The query relation $Q$.
\end{itemize}

It turns out that maintaining $F$ and $T$ is even easier than indicated in the above sketch since the computation that starts at time $t+1$ can reuse information computed by that at time $t$ and so forth.

Relation $F$ is very easy to \emph{maintain}. We recall that its first version (for time $t=0$) is given by the initialization. Whether a tuple $(w,u,v)$ is in $F$ \emph{after} change operation $\delta$ can be determined as follows: Let $\delta'$ be the last symbol (operation) of $\alpha_w$ and $\alpha'$ be its prefix of length $(\log n)-1$, hence $\alpha_w=\alpha'\delta'$. Let  $G'$ be the graph represented by all tuples of the form $(w_{\delta \alpha'},\cdot,\cdot)$. Then $(w,u,v)$ is in $F$ after applying $\delta$ if and only if $(u,v)$ is in $\delta'(G')$.

Relation $T$ can be maintained similarly.  A tuple $(0,w,\tpl a)$ is in $T$ if $\tpl a$ is in $\varphi_q(\emptyset)$ applied to the graph $F(w,\cdot, \cdot)$. A tuple $(i+1,w,\tpl a)$ is in $T$ after operation $\delta$, for $i<\log n$, if and only if the tuple $\tpl a$ is in the relation $\varphi_q(A)$, where $A$ consists of all tuples $\tpl b$, for which $(i,w_{\delta\alpha'},\tpl b)$ is in $T$ before operation $\delta$, where $\alpha'$ is as above. Here we assume, without loss of generality, that $\varphi_q(R)$ uses the graph relation only when $R$ is empty.

The query relation $Q$ can be obtained from the tuples of the form $(\log n,w,\tpl a)$ in $T$.
\end{proofof}

\begin{proofof}{Theorem~\ref{theorem:languages:EFO}}
We proceed in two steps. We first prove the result for quantifier-free change operations and show afterwards that the case of $\EFO$-definable change operations can be replaced  reduced by quantifier-free change operations.

For the first step, we simply adapt the proof of \cite[Theorem 4.1]{GeladeMS12}. 

Let $L$ be a context-free language over alphabet $\Sigma$ and $G$ a corresponding context-free grammar in Chomsky normal form, where we additionally allow rules of the form $X \rightarrow \epsilon$ and assume the existence of a non-terminal $E$ with $E \rightarrow \epsilon$ and $X \rightarrow XE, X \rightarrow EX$, for every non-terminal $X$.\footnote{This non-terminal is basically needed for deletions, see \cite{GeladeMS12}.}

In \cite{GeladeMS12}, auxiliary relations $S_{X,Y}(i_1, j_1, i_2, j_2)$ are used\footnote{Actually, the choice of indices is slightly different in \cite{GeladeMS12} but this difference is inessential.}, for each pair $X,Y$ of non-terminals of the grammar $G$, with the intention that $(i_1,j_1,i_2,j_2) \in S_{X,Y}$ if and only if $i_1 \leq j_1 < i_2 \leq j_2$ and from $X$ the sentential form $w_{i_1}\cdots w_{j_1}Yw_{i_2}\cdots w_{j_2}$ can be derived using rules of $G$.  In particular, if $(i_1,j_1,i_2,j_2) \in S_{X,Y}$ and $Y \Rightarrow^* w_{j_1+1}\cdots w_{i_2-1}$, then $X \Rightarrow^* w_{i_1}\cdots w_{j_2}$.

We adapt this approach in a similar fashion as in the proof of Theorem \ref{theorem:languages:DynProp}. To this end, we use auxiliary relations of the form $S^{f_1,f_2}_{X,Y}$ where $f_1, f_2$ are relabeling functions. The intension is that  $(i_1,j_1,i_2,j_2) \in S^{f_1,f_2}_{X,Y}$ if and only if $i_1 \leq j_1 < i_2 \leq j_2$ and $X \Rightarrow^* f_1(w_{i_1})\cdots f_1(w_{j_1})Yf_2(w_{i_2})\cdots f_2(w_{j_2})$.

We show next how to maintain these relations under quantifier-free replacement queries $\rho(p)$ with one parameter. As in the proof of Theorem~\ref{theorem:languages:DynProp}, the generalization for any number of parameters  is straightforward, but tedious.

Let $\state$ be a program state and $\delta=(\rho,a)$ a change operation. Similarly as in the proof of Theorem~\ref{theorem:languages:DynProp}, from $\rho$ and the symbol $\tau$ at position $a$, one can derive relabeling functions $f_\leftarrow$, $f_\rightarrow$ and a symbol $\sigma=\sigma(\tau)$ that represent the modifications applied to the input string.

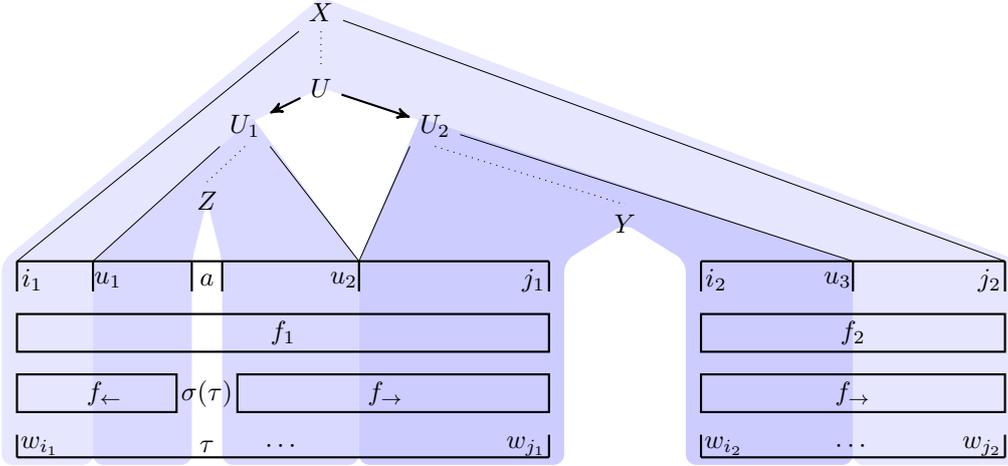
\begin{figure}
\begin{tikzpicture}[thick]
\pgfmathsetmacro{\by}{-0.3} 
\pgfmathsetmacro{\ty}{-0.55} 
\pgfmathsetmacro{\py}{-0.7} 
\fill [fill=blue!10,rounded corners] (-1,3.2) -- (-5.2,\by) -- (-5.2,-3) -- (-4,-3) --
(-4,\by) -- (-2,1.5) -- (-1,2) -- (0.5,1.5) -- (6,\by) -- (6,-3) --
(8.2,-3) -- (8.2,\by) -- (-1,3.2);
\fill [fill=blue!15,rounded corners] (-1.9,1.6) -- (-4,\by) -- (-4,-3) -- (-2.7,-3) --
(-2.7,\by) -- (-2.5,.5) -- (-2.3,\by) -- (-2.3,-3) -- (-0.5,-3) --
(-0.5,\by) -- (-1.9,1.6);
\fill [fill=blue!20,rounded corners] (0.3,1.6) -- (-0.5,\by) --
(-0.5,-3) -- (2.2,-3) -- (2.2,\by) -- (3,0.2) -- (3.8,\by) -- (3.8,-3) -- (6,-3)
-- (6,\by) -- (0.3,1.6);
\node (X) at (-1,3) {$X$};
\node (U1) at (-2,1.5) {$U_1$};
\node (U2) at (0.5,1.5) {$U_2$};
\node (U) at (-1,2) {$U$}
 edge[>=stealth', ->] (U1) edge[>=stealth', ->] (U2);
\node (Y) at (3,.2) {$Y$};
\node (Z) at (-2.5,.5) {$Z$};
\draw [very thin] (X.south west) -- (-5,\by);
\draw [very thin]  (X.340) -- (8,\by);
\draw [very thin]  (U1.south west) -- (-4,\by);
\draw [very thin]  (U1.south east) -- (-.5,\by);
\draw [very thin]  (U2.south west) -- (-.5,\by);
\draw [very thin]  (U2.340) -- (6,\by);
\draw [thin,dotted] (X.south) -- (U.north);
\draw [thin,dotted] (U1.south) -- (Z.north);
\draw [thin,dotted] (U2.south) -- (Y.north);
\draw(-5,\by) -- (2,\by);
\draw(4,\by) -- (8,\by);
\foreach \x in {-5,-4,-.5,6,8,-2.7,-2.3,2,4} {\draw (\x,\by) -- (\x,\py);} 
\node at (-4.8,\ty) {$i_1$};
\node at (-3.8,\ty) {$u_1$};
\node at (-2.5,\ty) {$a$};
\node at (-.7,\ty) {$u_2$};
\node at (1.8,\ty) {$j_1$};
\node at (4.2,\ty) {$i_2$};
\node at (5.8,\ty) {$u_3$};
\node at (7.8,\ty) {$j_2$};
\draw (-5,-1) rectangle (2,-1.5);
\draw (4,-1) rectangle (8,-1.5);
\node at (-1.5,-1.25) {$f_1$};
\node at (6,-1.25) {$f_2$};
\draw (-5,-1.8) rectangle (-2.9,-2.3);
\draw (-2.1,-1.8) rectangle (2,-2.3);
\draw (4,-1.8) rectangle (8,-2.3);
\node at (-3.85,-2.05) {$f_\leftarrow$};
\node at (-2.5,-2.05) {$\sigma(\tau)$};
\node at (-.15,-2.05) {$f_\rightarrow$};
\node at (6,-2.05) {$f_\rightarrow$};

\foreach \x in {-5,8,2,4} {\draw (\x,-2.6) -- (\x,-2.9);} 
\draw (-5,-2.9) -- (2,-2.9);
\draw (4,-2.9) -- (8,-2.9);
\node at (-4.7,-2.75) {$w_{i_1}$};
\node at (-2.5,-2.75) {$\tau$};
\node at (-1.5,-2.75) {$\cdots$};
\node at (1.7,-2.75) {$w_{j_1}$};
\node at (4.3,-2.75) {$w_{i_2}$};
\node at (6,-2.75) {$\cdots$};
\node at (7.7,-2.75) {$w_{j_2}$};
\end{tikzpicture}
\caption{Conditions for the update of $S^{f_1,f_2}_{X,Y}$ in the proof of Theorem \ref{theorem:languages:EFO} (see also \cite[Fig. 2]{GeladeMS12}). In the bottom layer there is the string before the modification. One layer above the changes of the modification are depicted, the next layer shows the transformation assumed by the auxiliary relation. On top, we see the partial derivation tree checked by the update formula. \label{fig:languages:tree}}
\end{figure}
We explain the part of the update formula of $S^{f_1,f_2}_{X,Y}$ that deals with the case~\mbox{$i_1 < a < j_1$}. The other cases are similar. Inside this case, we omit some border sub-cases.  
The update formula essentially checks whether for some non-terminals $U, U_1, U_2$ and $Z$ with $Z \rightarrow f_1(\sigma)$ and $U \rightarrow U_1U_2$ there are positions $u_1, u_2, u_3$ with $i_1 \leq u_1 < a < u_2 < j_1$ and $i_2 \leq u_3 \leq j_2$ such that 
$S^{f_1 \circ f_\leftarrow,f_2 \circ f_\rightarrow}_{X,U}(i_1,u_1-1,u_3+1,j_2)$, 
$S^{f_1 \circ f_\leftarrow,f_1 \circ f_\rightarrow}_{U_1,Z}(u_1,a-1,a+1,u_2)$, and 
$S^{f_1 \circ f_\rightarrow,f_2 \circ f_\rightarrow}_{U_2,Y}(u_2+1,j_1,i_2, u_3)$ hold.
An illustration is given in Figure \ref{fig:languages:tree}.
 
The initialization of the relations $S^{f_1,f_2}_{X,Y}$ is straightforward. If, for relabeling functions $f_1,f_2$ it holds $f_1(\epsilon)=\sigma_1$ and  $f_2(\epsilon)=\sigma_2$, then a tuple $(i_1,j_1,i_2,j_2)$ is in $S^{f_1,f_2}_{X,Y}$ if and only if $X\Rightarrow^* \sigma_1^{j_1-i_1+1}Y\sigma_2^{j_2-i_2+1}$.
\medskip

It remains to show how the case of  $\EFO$-definable change operations can be reduced to the quantifier-free case. More precisely, we show that for each $\EFO$-definable change operation $\rho(p)$ there exists a quantifier-free change operation $\rho'(p,\tpl{q})$ with many additional parameters collected in $\tpl{q}$ such that, for each string $w$ and each position $a$, there is a tuple $\tpl{c}$ of  positions such that $\delta(w)=\delta'(w)$, where $\delta=(\rho,a)$ and $\delta'=(\rho',(a,\tpl c))$. Since we already showed that  quantifier-free change operations of any arity can be handled in \DynFO, the theorem then follows.

To this end, let $\mu_\sigma(p;x)$ be a $\EFO$-formula which expresses whether after change operation $\rho(p)$, position $x$ will carry symbol $\sigma$. Without loss of generality, we can assume that  $\mu_\sigma(p;x)=\bigvee_{m} \theta_m$, where each $\theta_m(p;x)$ is of the form  $\exists y_1\ldots,y_k \psi_m$, for some $k$, and each $\psi_m(p,\tpl y;x)$ describes a full atomic type over $p, y_1\ldots,y_k,x$ with respect to the linear order and the letter relations. That is, $\psi_m$ completely specifies the relative order of the positions bound to the variables and the symbols they carry. 

We claim that, for every $m$, for each word $w$ and each position $a$ there exists a tuple $\tpl c=(c_1,\ldots,c_k)$ of positions in $w$ such that for every position $i$ it holds $w\models \theta_m(a,i)$ if and only if $w\models \psi_m(a,\tpl c;i)$.

We can assume without loss of generality that the order-type specified by $\psi_m$ fulfills $y_1<\cdots y_\ell<x<y_{\ell+1}<\cdots<y_k$ and $y_j=p$, for some $j$ and $\ell$. This can be achieved by adding or removing and renaming variables accordingly.
 
Then, if $w\models \theta_m(a,i)$ holds, positions $c_1,\ldots,c_\ell$ for variables $y_1,\ldots,y_\ell$ can be chosen such that the tuple $(c_1,\ldots,c_\ell)$ is lexicographically minimal and $c_{k},\ldots,c_{\ell+1}$ for $y_{k},\ldots,y_{\ell+1}$ such  that the tuple $(c_k,\ldots,c_{\ell+1})$ is lexicographically maximal. It is easy to see that this tuple $\tpl c$ fulfills the condition of the claim. We call $\tpl c$ the canonical tuple for $\theta_m$ with respect to $w$ and $a$.

The update formulas for $\rho(p)$ can now be obtained as follows. First, the update program for the replacement queries given by the formulas $\bigvee_m \psi_m(p,\tpl c_m;x)$, where all tuples $\tpl c_m$ are pairwise disjoint, is determined. In the resulting update formulas, a block of existential quantifiers for the tuples  $\tpl c_m$ is added and it is verified that each tuple is bound to a canonical tuple of positions.%
\end{proofof}

\begin{proofof}{Theorem~\ref{theorem:tc:general:insertion}, continued}
  
Both graphs $G$ and $G'$ are not acyclic and thus not
suitable for (b). Yet a slight modification of the above construction
yields acyclic graphs $G$ and $G'$ but it uses an existential
quantifier in the definition of the change operation. The graphs are
depicted in Figure \ref{figure:tc:general:insertionb}. The graph $G'$
is obtained by applying the operation $\delta = \mf{E}{x,y}{} \df
E(x,y) \lor \exists z \big(E(z, x) \land E(z, y)\big)$ to $G$. The
proof is now analogous to (a) except that $G$ has to be first-order
interpreted into the path graph $G_0$, as it uses a slightly larger
domain. The rest of the argument for (b) is
analogous, but $\Qreach(G_0)$ is obtained from $\Qreach(G')$ as 
$\Qreach(G_0)=\Qreach(G')\cap (V_0\times V_0)$.

  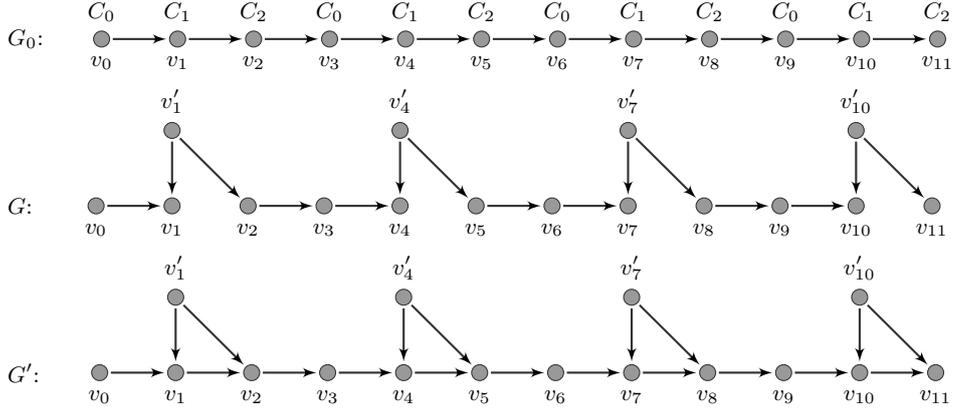
\begin{figure}[ht] 
\begin{tikzpicture}[
      xscale=1.0,
      yscale=1.0,
      font=\footnotesize,
    ]
          \node (tmp) at (-1,0) {$G_0$:};
          \node (0) at (0,0)[mnode, label=below:$v_0$, label=above:$C_0$] {};
          \node (1) at (1,0)[mnode, label=below:$v_1$, label=above:$C_1$] {};
          \node (2) at (2,0)[mnode, label=below:$v_2$, label=above:$C_2$] {};
          \node (3) at (3,0)[mnode, label=below:$v_3$, label=above:$C_0$] {};
          \node (4) at (4,0)[mnode, label=below:$v_4$, label=above:$C_1$] {};
          \node (5) at (5,0)[mnode, label=below:$v_5$, label=above:$C_2$] {};
          \node (6) at (6,0)[mnode, label=below:$v_6$, label=above:$C_0$] {};
          \node (7) at (7,0)[mnode, label=below:$v_7$, label=above:$C_1$] {};
          \node (8) at (8,0)[mnode, label=below:$v_8$, label=above:$C_2$] {};
          \node (9) at (9,0)[mnode, label=below:$v_9$, label=above:$C_0$] {};
          \node (10) at (10,0)[mnode, label=below:$v_{10}$, label=above:$C_1$] {};
          \node (11) at (11,0)[mnode, label=below:$v_{11}$, label=above:$C_2$] {};
 
           \draw [dEdge] (0) to (1);
           \draw [dEdge] (1) to (2);
           \draw [dEdge] (2) to (3);
           \draw [dEdge] (3) to (4);
           \draw [dEdge] (4) to (5);
           \draw [dEdge] (5) to (6);
           \draw [dEdge] (6) to (7);
           \draw [dEdge] (7) to (8);
           \draw [dEdge] (8) to (9);
           \draw [dEdge] (9) to (10);
           \draw [dEdge] (10) to (11);

    \end{tikzpicture}  
  
    \begin{tikzpicture}[
      xscale=1.0,
      yscale=1.0,
      font=\footnotesize,
    ]
          \node (tmp) at (-1,0) {$G$:};
          \node (0) at (0,0)[mnode, label=below:$v_0$] {};
          \node (1) at (1,0)[mnode, label=below:$v_1$] {};
          \node (1a) at (1,1)[mnode, label=above:$v'_1$] {};
          \node (2) at (2,0)[mnode, label=below:$v_2$] {};
          \node (3) at (3,0)[mnode, label=below:$v_3$] {};
          \node (4) at (4,0)[mnode, label=below:$v_4$] {};
          \node (4a) at (4,1)[mnode, label=above:$v'_4$] {};
          \node (5) at (5,0)[mnode, label=below:$v_5$] {};
          \node (6) at (6,0)[mnode, label=below:$v_6$] {};
          \node (7) at (7,0)[mnode, label=below:$v_7$] {};
          \node (7a) at (7,1)[mnode, label=above:$v'_7$] {};
          \node (8) at (8,0)[mnode, label=below:$v_8$] {};
          \node (9) at (9,0)[mnode, label=below:$v_9$] {};
          \node (10) at (10,0)[mnode, label=below:$v_{10}$] {};
          \node (10a) at (10,1)[mnode, label=above:$v'_{10}$] {};
          \node (11) at (11,0)[mnode, label=below:$v_{11}$] {};
 
           \draw [dEdge] (0) to (1);
           \draw [dEdge] (1a) to (1);
           \draw [dEdge] (1a) to (2);           
           \draw [dEdge] (2) to (3);
           \draw [dEdge] (3) to (4);
           \draw [dEdge] (4a) to (4);
           \draw [dEdge] (4a) to (5);
           \draw [dEdge] (5) to (6);
           \draw [dEdge] (6) to (7);
           \draw [dEdge] (7a) to (7);
           \draw [dEdge] (7a) to (8);
           \draw [dEdge] (8) to (9);
           \draw [dEdge] (9) to (10);
           \draw [dEdge] (10a) to (10);
           \draw [dEdge] (10a) to (11);
                      
    \end{tikzpicture}

        \begin{tikzpicture}[
      xscale=1.0,
      yscale=1.0,
      font=\footnotesize
    ]
          \node (tmp) at (-1,0) {$G'$:};
          \node (0) at (0,0)[mnode, label=below:$v_0$] {};
          \node (1) at (1,0)[mnode, label=below:$v_1$] {};
          \node (1a) at (1,1)[mnode, label=above:$v'_1$] {};
          \node (2) at (2,0)[mnode, label=below:$v_2$] {};
          \node (3) at (3,0)[mnode, label=below:$v_3$] {};
          \node (4) at (4,0)[mnode, label=below:$v_4$] {};
          \node (4a) at (4,1)[mnode, label=above:$v'_4$] {};
          \node (5) at (5,0)[mnode, label=below:$v_5$] {};
          \node (6) at (6,0)[mnode, label=below:$v_6$] {};
          \node (7) at (7,0)[mnode, label=below:$v_7$] {};
          \node (7a) at (7,1)[mnode, label=above:$v'_7$] {};
          \node (8) at (8,0)[mnode, label=below:$v_8$] {};
          \node (9) at (9,0)[mnode, label=below:$v_9$] {};
          \node (10) at (10,0)[mnode, label=below:$v_{10}$] {};
          \node (10a) at (10,1)[mnode, label=above:$v'_{10}$] {};
          \node (11) at (11,0)[mnode, label=below:$v_{11}$] {};
 
           \draw [dEdge] (0) to (1);
           \draw [dEdge] (1a) to (1);
           \draw [dEdge] (1a) to (2);           
           \draw [dEdge] (1) to (2);
           \draw [dEdge] (2) to (3);
           \draw [dEdge] (3) to (4);
           \draw [dEdge] (4a) to (4);
           \draw [dEdge] (4a) to (5);
           \draw [dEdge] (4) to (5);
           \draw [dEdge] (5) to (6);
           \draw [dEdge] (6) to (7);
           \draw [dEdge] (7a) to (7);
           \draw [dEdge] (7a) to (8);
           \draw [dEdge] (7) to (8);
           \draw [dEdge] (8) to (9);
           \draw [dEdge] (9) to (10);
           \draw [dEdge] (10a) to (10);
           \draw [dEdge] (10a) to (11);
           \draw [dEdge] (10) to (11);
    \end{tikzpicture}

    \caption{The graphs from the proof of Theorem \ref{theorem:tc:general:insertion}(b).}
      \label{figure:tc:general:insertionb}
      
  \end{figure}

This completes the proof for \DynFO. For the full statement, with
ordered graphs, it suffices to replace Lemma \ref{lemma:tc:pathgraphs} by the following Lemma~\ref{lemma:lower-reach-order}.
\end{proofof}

\begin{lemma}\label{lemma:lower-reach-order}
  For every first-order formula $\varphi(x,y)$ over a signature $\sigma=\{E,<,A_1,\ldots,A_m\}$, where $A_1,\ldots,A_m$ are unary, there is a graph $G$, consisting of a single path, and a linear order $<$ on its vertices, such that for all sets $A_1,\ldots,A_m$ of nodes of $G$, there are nodes $u,v_1,v_2$ such that $(G,<,A_1,\ldots,A_m)\models \varphi(u,v_1) \liff \varphi(u,v_2)$, but $v_2$ is reachable from $u$ and $v_1$ is not.
\end{lemma}
\begin{proof}
Let $\varphi(x,y)$ be a formula as stated in the lemma and let $k$ be
its quantifier rank. Let $\tau$ be the signature $\sigma\cup
\{B_1,\ldots,B_{r}\}$ with $r=2^{k+1}+1$, and let $\ell$ be the number
of $\FO[k,1]$-types of $\tau$-structures. The 
additional relations $B_1,\ldots,B_{r}$ will be useful later in the
proof, when
the so-called Extension Theorem is applied.

Let $G$ be the graph with vertex set $V=\{(i,j)\mid 0\le i\le r, 1\le j\le \ell+1
\}$ and edges between $(i_1,j_1)$ and $(i_2,j_2)$ if
$i_1+1=i_2$ and $j_1=j_2$ or if $i_1=r$, $i_2=0$, and
$j_1+1=j_2$. That is, $G$ consists of a path with $(\ell+1) (r+1)$
nodes from $(0,1)$ to $(r,\ell+1)$ that consists of  $\ell+1$
sub-paths of length $r$ from nodes $(0,i)$ to $(r,i)$, which are
connected by edges of the form $((r,i),(0,i+1))$. We refer to $i$ in
$(i,j)$ as its \emph{column number} and to $j$ as its \emph{row
  number} and to \emph{columns} and \emph{rows} of the graph,
accordingly. This graph is depicted in Figure \ref{figure:tc:general:order}.

  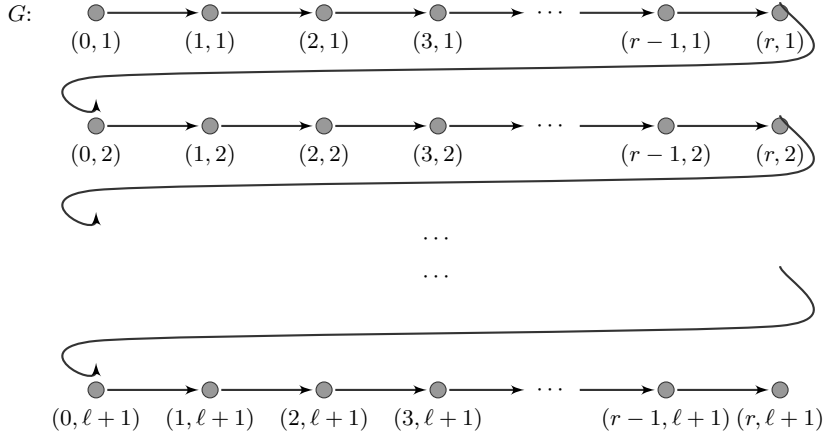
\begin{figure}[ht] 
      \begin{tikzpicture}[
      xscale=1.0,
      yscale=1.0,
      font=\footnotesize
    ]
          \node (tmp) at (-1,0) {$G$:};
          \node (0) at (0,0)[mnode, label=below:{$(0,1)$}] {};
          \node (1) at (1.5,0)[mnode, label=below:{$(1,1)$}] {};
          \node (2) at (3,0)[mnode, label=below:{$(2,1)$}] {};
          \node (3) at (4.5,0)[mnode, label=below:{$(3,1)$}] {};
          \node (4) at (6,0) {$\cdots$};
          \node (5) at (7.5,0)[mnode, label=below:{$(r-1,1)$}] {};
          \node (6) at (9,0)[mnode, label=below:{$(r,1)$}] {};
          \node (10) at (0,-1.5)[mnode, label=below:{$(0,2)$}] {};
          \node (11) at (1.5,-1.5)[mnode, label=below:{$(1,2)$}] {};
          \node (12) at (3,-1.5)[mnode, label=below:{$(2,2)$}] {};
          \node (13) at (4.5,-1.5)[mnode, label=below:{$(3,2)$}] {};
          \node (14) at (6,-1.5) {$\cdots$};
          \node (15) at (7.5,-1.5)[mnode, label=below:{$(r-1,2)$}] {};
          \node (16) at (9,-1.5)[mnode, label=below:{$(r,2)$}] {};
          \node (20) at (0,-3){};
          \node (23) at (4.5,-3) {$\cdots$};
          \node (33) at (4.5,-3.5) {$\cdots$};
          \node (36) at (9,-3.5){};
          \node (40) at (0,-5)[mnode, label=below:{$(0,\ell+1)$}] {};
          \node (41) at (1.5,-5)[mnode, label=below:{$(1,\ell+1)$}] {};
          \node (42) at (3,-5)[mnode, label=below:{$(2,\ell+1)$}] {};
          \node (43) at (4.5,-5)[mnode, label=below:{$(3,\ell+1)$}] {};
          \node (44) at (6,-5) {$\cdots$};
          \node (45) at (7.5,-5)[mnode, label=below:{$(r-1,\ell+1)$}] {};
          \node (46) at (9,-5)[mnode, label=below:{$(r,\ell+1)$}] {};
 
           \draw [dEdge] (0) to (1);        
           \draw [dEdge] (1) to (2);
           \draw [dEdge] (2) to (3);
           \draw [dEdge] (3) to (4);
           \draw [dEdge] (4) to (5);
           \draw [dEdge] (5) to (6);
           \draw [dEdge,  shorten >=4pt, shorten <=4pt] plot [smooth, tension=0.4] coordinates { (6) ($ (6) +(0,-0.65)$) ($ (10) +(-0,0.65)$) (10)};
           \draw [dEdge] (10) to (11);        
           \draw [dEdge] (11) to (12);
           \draw [dEdge] (12) to (13);
           \draw [dEdge] (13) to (14);
           \draw [dEdge] (14) to (15);
           \draw [dEdge] (15) to (16);           
            \draw [dEdge,  shorten >=4pt, shorten <=4pt] plot [smooth, tension=0.4] coordinates { (16) ($ (16) +(0,-0.65)$) ($ (20) +(-0,0.65)$) (20)};
           \draw [dEdge,  shorten >=4pt, shorten <=4pt] plot [smooth, tension=0.4] coordinates { (36) ($ (36) +(0,-0.65)$) ($ (40) +(-0,0.65)$) (40)}; 
           \draw [dEdge] (40) to (41);        
           \draw [dEdge] (41) to (42);
           \draw [dEdge] (42) to (43);
           \draw [dEdge] (43) to (44);
           \draw [dEdge] (44) to (45);
           \draw [dEdge] (45) to (46);
    \end{tikzpicture}

    \caption{The graph from the proof of Lemma \ref{lemma:lower-reach-order}.}
      \label{figure:tc:general:order}
      
  \end{figure}

Let $<$ be just the lexicographic order (or stated otherwise: column-major order) on $V$. 

Let now $A_1,\ldots,A_m$ be arbitrary sets of nodes of $G$. Let
$\calB$ be the $\tau$-structure that is obtained from
$(G,<,A_1,\ldots,A_m)$ by removing column 0 and its adjacent edges,
and by adding the unary relations $B_1,\ldots,B_r$, where
each set $B_i$ is just the set of nodes of column $i$.
Since $G$ has more rows than there are $\FO[k,1]$-types of $\tau$-structures, there must be two nodes $v_1\df (2^k+1,i)$ and $v_2\df (2^k+1,j)$, $i<j$, which have the same $\FO[k,1]$-type in the $\tau$-structure $\calB$.

Clearly, there is a path from $u\df (0,j)$ to $v_2$ in $G$, but not from $u$ to $v_1$. However, we will show in the following that $(G,<,A_1,\ldots,A_m)\models \varphi(u,v_1) \liff \varphi(u,v_2)$. To this end, we show that the duplicator has a winning strategy in the $k$-round Ehrenfeucht game on the two structures $(G,<,A_1,\ldots,A_m,u,v_1)$ and $(G,<,A_1,\ldots,A_m,u,v_2)$. This follows with the help of the Extension Theorem (Theorem 8) from \cite{Schwentick96}, as we explain next. For the convenience of readers, we repeat  it as Theorem~\ref{theorem:extension} below.

In a nutshell, the Extension Theorem guarantees the existence of a
winning strategy for the duplicator by combining two strategies. To
explain the first strategy, let
$H$ denote column $2^k+1$ of $\calB$ with the distinguished
element $v_1$ and $H'$ the same column but with $v_2$ as distinguished element. 
The
first winning strategy is for the game on the neighborhoods of diameter
$2^k$ of $H$ and $H'$, that is on $(\calB,v_1)$ and
$(\calB,v_2)$. Such a winning strategy  exists, because $v_1$ and $v_2$ have the same
$\FO[k,1]$-type in $\calB$. Thanks to the $B_i$-relations this
strategy has the additional property that the duplicator answers each
move of the spoiler by a move (node) in the same column. 
The second winning strategy is trivial: it
is for the game on two identical copies of the structure obtained from
$(G,<,A_1,\ldots,A_m,u)$ by removing column $H$. The Extension Theorem
allows to combine these two strategies, thanks to the additional
property and because the linear order is very homogeneous with respect
to the column structure.

As signature $S$ in the application of the Extension Theorem we choose
$\sigma\cup \{c_1,c_2\}$. The structures $\calA$ and $\calA'$ are
$(G,<,A_1,\ldots,A_m,u,v_1)$ and $(G,<,A_1,\ldots,A_m,u,v_2)$,
respectively. The distance function is just the distance with respect
to $E$, in both structures. The sets $H$ $H'$ are chosen as above (and
there are no other $H_i$ or $H'_i$ required). Clearly, if we remove $H$ and $H'$ the resulting
structures $\calA-H$ and $\calA'-H$ are isomorphic via the identity
mapping~$\id$. To apply the Extension Theorem it remains to verify
that its conditions (i)--(iii) are fulfilled. 

Condition (i) just states that the duplicator has a winning structure
on $\calB$ and $\calB'$ that has the additional, ``distance from
$H$''-respecting property. It holds as explained above.

Condition (ii) holds, because the structures are identical and thus
$\id$ is the required isomorphism.

Condition (iii) is more complicated. For nodes $x_1,x_2$ from $\calA$ and $x'_1,x'_2$ from $\calA'$, let ($*$) be the condition that

$(x_1,x_2)\in E$ if and only if $(x'_1,x'_2)\in E$, and $x_1<x_2$ if and only if   $x'_1<x'_2$.

To establish Condition (iii), we have to show that ($*$) holds  under
the assumption that there exists a number $e< 2^k$ such that the
following statements hold.
\begin{enumerate}[(a)]
\item None of $x_1,x_2,x'_1,x'_2$ has distance $e$ from $H$ (or $H'$). (Henceforth, we call nodes whose distance is smaller than $e$ \emph{inner nodes} and the others \emph{outer nodes}).
\item If $x_1$ is an inner node then $x'_1$ has the same distance from $H'$. Otherwise, even $x_1=x'_1$. And likewise for $x_2$ and $x'_2$.
\item For the inner nodes from $x_1,x_2$ the (joint) atomic $\sigma$-type is the same as for the respective nodes from $x'_1,x'_2$.
\end{enumerate}
To conclude ($*$) from (a)--(c), we distinguish three cases: If $x_1$ and $x_2$ are both inner nodes, then ($*$) follows from (c). If both are outer nodes, it follows from (b). Finally, if $x_1$ is an inner node and $x_2$ an outer node then (a) guarantees that neither $(x_1,x_2)$ nor $(x'_1,x'_2)$ are edges and the relations $B_i$ guarantee that  $x_1<x_2$ if and only if   $x'_1<x'_2$.

This concludes the proof of the lemma.
\end{proof}

\begin{theorem}[Extension Theorem, Theorem 8 in \cite{Schwentick96}]\label{theorem:extension}
  Let $k>0$.
Let $S$ be a signature with relational symbols $R_1,\ldots,R_s$ of arities $a_1,\ldots,a_s$ and constant symbols $c_1,\ldots,c_t$.

Let $\calA,\calA'$ be $S$-structures. Let $\delta,\delta'$ be distance functions on $\calA$ and $\calA'$, respectively.

Let $H_1,\ldots,H_l$ and $H'_1,\ldots,H'_l$ be sequences of subsets of $U^{\calA}$ and $U^{\calA'}$ respectively, such that $N^{\delta}_{2^k}(H_i) \cap N^{\delta}_{2^k}(H_j) = \emptyset$ and $N^{\delta'}_{2^k}(H'_i) \cap N^{\delta'}_{2^k}(H'_j) = \emptyset$ for $i \not = j$.

Let $\alpha$ be an isomorphism from $\calA \downarrow (U^{\calA}-(H_1 \cup \cdots \cup H_l))$ to\\
$\calA'\downarrow (U^{\calA'}-(H'_1 \cup \cdots \cup H'_l))$.

Let $\delta_j(x):= \delta(x,H_j)$ and $\delta'_j(x):= \delta'(x,H'_j)$.

{\bf Duplicator has a winning strategy in the $k$-round FO Ehrenfeucht  game on $\calA$ and $\calA'$, if the following conditions are fulfiled.}
\begin{itemize}
\item[(i)]
For every $j \le l$, Duplicator has a winning strategy in the $k$-round Ehrenfeucht  game on
$(\calA \downarrow N_{2^k}^{\delta}(H_j),\delta_j)$ and $(\calA' \downarrow N_{2^k}^{\delta'}(H'_j),\delta'_j)$. 
\item[(ii)]
For every $x \in U^{\calA}-(H_1 \cup \cdots \cup H_l)$ and every $j \le l$
\[
\delta_j(x) = \delta'_j(\alpha(x)) \qquad \mbox{or} \qquad (\delta_j(x)>2^k \; \mbox{and} \; \delta'_j(\alpha(x))>2^k).
\]
\item[(iii)]
For every $p \le s$, all sequences $x_1,\ldots,x_{a_p}\in U^{\calA}$ and $x'_1,\ldots,x'_{a_p}\in U^{\calA'}$ it holds that $R_p^{\calA}(x_1,\ldots,x_{a_p}) \Longleftrightarrow R_p^{\calA'}(x'_1,\ldots,x'_{a_p})$ if there is $e < 2^k$ such that
\begin{itemize}
\item[(a)]
for every $i \le a_p$ and $j \le l$:  $\delta_j(x_i) \not = e+1 \not = \delta'_j(x'_i)$,
\item[(b)]
for every $i \le a_p$:
\begin{tabular}[t]{l}
if $\delta_j(x_i) \le e$ for some $j \le l$,  then $\delta_j(x_i) = \delta'_j(x'_i)$,\\  otherwise $\alpha(x_i)= x'_i$,
\end{tabular}
\item[(c)]
for every $j \le l$: if $x_{j_1},\ldots,x_{j_q}$ are exactly the elements of $x_1,\ldots,x_{a_p}$ in  $N_e^{\delta}(H_j)$ then $\calA \downarrow [x_{j_1},\ldots,x_{j_q}]  \cong \calA' \downarrow [x'_{j_1},\ldots,x'_{j_q}] $.
\end{itemize}
\end{itemize}

\end{theorem}

The notation $(\calA \downarrow N_{2^k}^{\delta}(H_j),\delta_j)$ in condition (i) is an abbreviation for a structure  that encodes the $\{0,\ldots, 2^k\}$ valued function $\delta_j$ by using $2^k+1$ fresh unary relations.

\begin{proofof}{Theorem~\ref{theorem:languages:EAFO}}
Let $L$ be the regular language $L((b^*ab^*ab^*)^*)$ over alphabet $\{a,b\}$ that contains all strings with an even number of $a$'s. This language can be maintained in \DynProp (and hence \DynFO) when only one position can change at a time. Indeed, already one nullary auxiliary relation (so, an auxiliary bit) suffices, namely the query relation~$Q$: every time an $a$ is inserted or deleted the bit $Q$ is flipped.

Let $\rho$ be the parameter-free replacement query that changes the label of a position $p$ from $b$ to $a$ if the position $p-1$ carries an $a$.
We show that $L$ cannot be maintained in \DynFO with only auxiliary
bits under single-tuple changes and $\rho$. To obtain a contradiction,
we assume there is a dynamic program $\prog$ with $m$ auxiliary bits and maximum quantifier depth $k$ of update formulas that achieves this. 

We consider strings of the form $(ab)^{n_1} (abb)^{n_2} \cdots (ab^{m+1})^{n_{m+1}}$ with $n_i \in \N$. For $I \subseteq \{1, \ldots, m+1\}$, let $w(I)$ be the string of this form such that $n_i = 2^k$ if $i \in I$ and $n_i = 2^k+1$ otherwise.
There are $2^m$ different valuations of $m$ auxiliary bits and $2^{m+1}$ different strings~$w(I)$, so let $I_1, I_2 \subseteq \{1,\ldots,m+1\}$ be different index sets such that in the states $\state_1 = (w(I_1),\aux)$ and $\state_2 = (w(I_2), \aux)$ reached by $\prog$ using only single insertions from an initially empty string, the auxiliary bits are valuated equally.

We show that $\prog$ cannot maintain the query for these instances under  $\rho$. Our strategy is as follows: the two strings will not be distinguishable from each other by first-order formulas with $k$ quantifiers after each change, so the auxiliary bits will be the same after each update. 
But at some point, the query answer will differ for the modified strings: by applying $\rho$, substrings of the form $b^*$ get smaller and eventually disappear. If an even number of $b^*$-substrings disappears, the parity of the number of $a$'s that are inserted by $\rho$ stays the same: if an even number of $a$'s are inserted by $\rho$, also the next change inserts an even number of $a$'s. If an odd number of $b^*$-substrings disappears, this parity changes.
By the choice of the two strings, at some point in one string an even number of $b^*$-substrings disappears, while in the other string an odd number of $b^*$-substrings disappears. But as the auxiliary bits and hence the query relation are the same at all time, the answer after the next change will be wrong for one string.

We first prove that for every number $\ell$ of times we apply $\rho$ to the strings, (1) $\prog$ will assign the same auxiliary bits to both strings, so for the states $(\updateDB{\rho^\ell}{w(I_1)}, \aux_1) = \updateState{P}{\rho^\ell}{\state_1}$ and $(\updateDB{\rho^\ell}{w(I_2)}, \aux_2) = \updateState{P}{\rho^\ell}{\state_2}$ for the update program $P$ of $\prog$, it holds that $\aux_1 = \aux_2$; and (2) $\updateDB{\rho^\ell}{w(I_1)} \equiv_k \updateDB{\rho^\ell}{w(I_2)}$, so the two strings cannot be distinguished by first-order formulas of quantifier depth $k$.
By choice of $I_1$ and $I_2$ (1) is the case for $\ell=0$. For (2), observe that if $w(I_1)$ and $w(I_2)$ have a different number of substrings of the form $ab^j$, then both numbers $n_j^1$ and $n_j^2$ are at least $2^k$. By a standard Ehrenfeucht-Fra\"isse argument, these numbers cannot be distinguished by first-order formulas of quantifier depth $k$. There are no other differences that could help differentiating $w(I_1)$ and $w(I_2)$.
For the inductive step, observe that (1) has to be true after applying
$\rho$ for the $\ell$-th time, as the update formulas can only access
the auxiliary bits and the strings after $\ell-1$ changes, but the
former are equal and the latter are not distinguishable by the
first-order update formulas of quantifier depth $k$. Therefore, the result of applying the update formulas has to be equal for every auxiliary bit. (2) is true with an analogous argumentation as in the base case.

It remains to prove that after some number of applications of $\rho$, one modified string is in $L$ while the other is not. As $\prog$ gives the same answer for both instances after each change, this shows that $\prog$ is not correct.

Assume that $w(I_1)$ and $w(I_2)$ are either both in $L$ or both not on $L$. Otherwise, $\prog$ is not correct and we are done.
Let $i$ be smallest index on which $I_1$ and $I_2$ differ. Without loss of generality, $i \in I_1$ and $i \notin I_2$.
Assume further that after each of the first $i$ applications of $\rho$, either both modified strings are in $L$, or both are not. Again, otherwise we are done.
It follows that before $\rho$ has been applied for the $i$-th time, the number of substrings $ab$ was either even for both strings or odd for both strings.
By applying $\rho$ for the $i$-th time, in the first string an even number of substring $a^iba$ are modified to $a^{i+1}$, so an even number of substrings $ab$ is removed, while in the second string an odd number of substrings $ab$ is removed.
So, when $\rho$ is applied for the $i+1$-th time, the number of positions whose label changes from $b$ to $a$ has for the first string the same parity as that number from the application before, while for the second string the parity of that number changes. As we assumed that after $i$ applications of $\rho$ both strings are in $L$ or both are not, after $i+1$ applications of $\rho$ exactly one of the modified strings is in $L$.
\end{proofof}

\end{document}